\documentclass[runningheads,dvipsnames]{llncs}

\usepackage{xcolor}
\usepackage{xparse}
\usepackage{booktabs}   
\usepackage{subcaption} 
\usepackage{amsmath,amsfonts}

\usepackage[bookmarks,unicode,colorlinks=true]{hyperref}

\usepackage[capitalize]{cleveref}
\makeatletter
   \def\@citecolor{blue}%
   \def\@urlcolor{blue}%
   \def\@linkcolor{blue}%

\def\orcidID#1{\smash{\href{http://orcid.org/#1}{\protect\raisebox{-1.25pt}{\protect\includegraphics{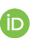}}}}}
\makeatother

\usepackage{footnote}
\makesavenoteenv{tabular}
\usepackage{thm-restate}
\usepackage{stmaryrd} 
\usepackage{centernot}
\usepackage{cmll}
\usepackage{xspace}
\usepackage{multicol}
\usepackage[inline]{enumitem}
\usepackage{rotating}
\usepackage{prooftree} 
\usepackage{mathpartir}
\usepackage{tikz}
\usepackage{xcolor}
\usepackage{listings}

\newcommand{\typec}[1]{{\color{RoyalBlue}{#1}}} 
\newcommand{\termc}[1]{{\color{RedViolet}{#1}}}
\newcommand{\kindc}[1]{{\color{RedOrange}{#1}}} 
\newcommand{\tsymc}[1]{{\color{OliveGreen}{#1}}} 
\newcommand{\nsymc}[1]{{\color{BrickRed}{#1}}} 
\newcommand{\blk}[1]{{\color{black}{#1}}}

\newcommand{\TT}{\typec{T}}
\newcommand{\UT}{\typec{U}}
\newcommand{\VT}{\typec{V}}

\newcommand{\XT}{\typec X}

\newcommand{\Xnt}{\nsymc X}
\newcommand{\Ynt}{\nsymc Y}
\newcommand{\Znt}{\nsymc Z}
\newcommand{\Ent}{\nsymc E}
\newcommand{\Fnt}{\nsymc F}
\newcommand{\Rnt}{\nsymc R}
\newcommand{\Ant}{\nsymc A}
\newcommand{\Bnt}{\nsymc B}
\newcommand{\ntdelta}{\nsymc\delta}

\newcommand{\tiota}{\typec \iota}
\newcommand{\talpha}{\typec \alpha}
\newcommand{\tbeta}{\typec \beta}

\newcommand{\grmeq}{\; ::= \;}
\newcommand{\grmor}{\; \mid \;}

\newcommand{\keyword}[1]{\ensuremath{\mathsf{#1}}\xspace}
\newcommand{\tkeyword}[1]{\keyword{\typec{#1}}}
\newcommand{\tekeyword}[1]{\keyword{\termc{#1}}}

\newcommand{\symkeyword}[1]{\keyword{\tsymc{#1}}}

\newcommand{\Label}[1]{\tkeyword{#1}}

\newcommand{\quitl}{\Label{Quit}}
\newcommand{\gol}{\Label{Go}}
\newcommand{\quits}{\symkeyword{Quit}}
\newcommand{\gos}{\symkeyword{Go}}

\newcommand{\nodel}{\Label{Node}}
\newcommand{\donel}{\Label{Done}}
\newcommand{\morel}{\Label{More}}
\newcommand{\dones}{\symkeyword{Done}}
\newcommand{\mores}{\symkeyword{More}}
\newcommand{\leafl}{\Label{Leaf}}

\newcommand{\fstl}{\Label{Fst}}
\newcommand{\sndl}{\Label{Snd}}

\newcommand{\intstreamt}{\tkeyword{IntStream}}
\newcommand{\streamt}{\tkeyword{Stream}}

\newcommand{\operator}[1]{\operatorname{#1}}
\newcommand{\fv}{\operator{fv}}      

\newcommand{\Eq}{\doteq} 
\newcommand{\Empty}{\varepsilon} 
\newcommand{\emptyword}{\nsymc\Empty}

\newcommand{\gequiv}{\approx} 
\newcommand{\Oplus}{\hspace{-.1ex}\oplus\hspace{-.1ex}} 

\newcommand{\nbb}{\mathbb{N}}

\newcommand{\E}{\mathcal{E}}

\newcommand{\N}{\mathcal{N}}

\newcommand{\R}{\mathcal{R}}

\newcommand{\T}{\mathcal{T}}

\newcommand{\X}{\mathcal X} 

\newcommand{\arity}{\operatorname{arity}}

\newcommand{\rename}{\operatorname{rename}}
\newcommand{\first}{\operatorname{first}}
\newcommand{\renameb}[1]{\blk{\rename(}#1\blk{)}}
\newcommand{\word}{\operatorname{word}}
\newcommand{\wordb}[1]{\blk{\word(}#1\blk{)}}
\newcommand{\tstate}{\operatorname{state}}
\newcommand{\tstateb}[1]{\blk{\tstate(}#1\blk{)}}
\newcommand{\subject}{\operatorname{subj}}
\newcommand{\subj}[1]{\subject(\termc{#1})}
\newcommand{\agree}[4]{\operatorname{agree}^{\termc{#1 #2}}(\termc{#3},\termc{#4})}

\newcommand{\declrel}[2]{#1\hfill\fbox{{#2}}}

\newcommand{\kind}{\kindc{\kappa}}
\newcommand{\kindp}{\kindc{\kappa'}}
\newcommand{\skind}{\kindc{\textsc s}}
\newcommand{\tkind}{\kindc{\textsc t}}

\newcommand{\karrow}[2]{\kindc{{#1}\Rightarrow{#2}}}
\newcommand{\kast}{\kindc *}

\newcommand{\End}{\tkeyword{End}}
\newcommand{\Skip}{\tkeyword{Skip}}
\newcommand{\Int}{\tkeyword{Int}}

\newcommand{\Bool}{\tkeyword{Bool}}
\newcommand{\Unit}{\tkeyword{Unit}}

\newcommand{\Dual}{\tkeyword{Dual}}
\newcommand{\Sharp}{\typec\sharp}
\newcommand{\Semi}{\typec;}
\newcommand{\Arrow}{\typec\rightarrow}

\newcommand{\INp}{?}
\newcommand{\INn}[1]{\typec{\INp{#1}}}		
\newcommand{\OUTp}{!}
\newcommand{\OUTn}[1]{\typec{\OUTp{#1}}}	
\newcommand{\MSGp}{\sharp}
\newcommand{\MSGn}[1]{\typec{\MSGp{#1}}}	
\newcommand{\IN}[2]{\typec{\INn{#1}.{#2}}}
\newcommand{\OUT}[2]{\typec{\OUTn{#1}.{#2}}}
\newcommand{\MSG}[2]{\typec{\MSGn{#1}.{#2}}}

\newcommand{\variants}[1]{\langle{#1}\rangle}
\newcommand{\recordf}[3]{\{{#1}\colon {#2}\}_{{#1}\in{#3}}} 
\newcommand{\records}[1]{\{#1\}}
\newcommand{\varrecf}[3]{\llparenthesis{#1}\colon {#2}\rrparenthesis_{{#1}\in{#3}}}
\newcommand{\varrecs}[1]{\llparenthesis{#1}\rrparenthesis}

\newcommand{\intchoicep}{\Oplus}
\newcommand{\intchoice}{\typec\intchoicep} 
\newcommand{\extchoicep}{\&}
\newcommand{\extchoice}{\typec\extchoicep}	
\newcommand{\choicep}{\odot}
\newcommand{\choice}{\typec\choicep}
\newcommand{\semit}[2]{\typec{{#1};{#2}}}
\newcommand{\function}[2]{\typec{#1 \rightarrow #2}}
\newcommand{\pairt}[2]{\typec{#1 \otimes #2}}
\newcommand{\tapp}[2]{\typec{{#1}\,{#2}}}
\newcommand{\tabs}[3]{\typec{\lambda\tbind{#1}{#2}.{#3}}}
\newcommand{\tlambda}[3]{\typec{\lambda\tbind{#1}{#2}.{#3}}}
\newcommand{\tforall}[1]{\typec{\forall}_\kindc{#1}}
\newcommand{\tmuinfix}[3]{\typec{\mu}\,\tbind{#1}{#2}\typec{.\,#3}}
\newcommand{\tmu}[1]{\typec{\mu}_\kindc{#1}}
\newcommand{\tbind}[2]{\typec{#1}\blk{\colon}\kindc{#2}} 

\newcommand{\foralltinfix}[3]{\typec{\forall\tbind{#1}{#2}.\,#3}}

\newcommand{\trecord}[2]{\typec{\{\overline{{#1}\colon{#2}}\}}}
\newcommand{\tvariant}[2]{\typec{\langle\overline{{#1}\colon{#2}}\rangle}}
\newcommand{\tbananas}[2]{\typec{\llparenthesis\overline{{#1}\colon{#2}}\rrparenthesis}}
\newcommand{\tchoice}[2]{\typec\odot\trecord{#1}{#2}}

\newcommand{\tinttchoice}[2]{\typec\oplus\trecord{#1}{#2}}

\newcommand{\tdual}[1]{\tapp\Dual{#1}}

\newcommand{\callt}[2]{\typec{{#1}\langle{#2}\rangle}}

\newcommand{\letk}{\tekeyword{let}}
\newcommand{\ink}{\tekeyword{in}}
\newcommand{\matchk}{\tekeyword{match}}
\newcommand{\withk}{\tekeyword{with}}
\newcommand{\casek}{\tekeyword{case}}
\newcommand{\ofk}{\tekeyword{of}}
\newcommand{\ask}{\tekeyword{as}}
\newcommand{\sendk}{\tekeyword{send}}
\newcommand{\receivek}{\tekeyword{receive}}
\newcommand{\forkk}{\tekeyword{fork}}
\newcommand{\newk}{\tekeyword{new}}

\newcommand{\closek}{\tekeyword{close}}
\newcommand{\reck}{\tekeyword{rec}}
\newcommand{\selectk}{\tekeyword{select}}
\newcommand{\selectc}[2]{\selectk\,\termc{#1}\,\ask\,\typec{#2}}

\newcommand{\ebind}[2]{\termc{#1}\colon\typec{#2}}
\newcommand{\eubind}[2]{\termc{#1}\colon^{\!\!\!\omega}\:\typec{#2}} 
\newcommand{\eunit}{\termc{\{\}}}
\newcommand{\eapp}[2]{\termc{#1}\,\termc{#2}}
\newcommand{\erec}[3]{\termc{\reck\,\ebind{#1}{#2}.{#3}}}
\newcommand{\eabs}[3]{\termc{\lambda\ebind{#1}{#2}.{#3}}}
\newcommand{\etabs}[3]{\termc{\Lambda\tbind{#1}{#2}.{#3}}} 
\newcommand{\etapp}[2]{\termc{{#1}[\typec{#2}]}}
\newcommand{\erecord}[2]{\termc{\{\overline{{#1}={#2}}\}}} 
\newcommand{\elet}[4]{\termc{\letk\,\{\overline{{#1}={#2}}\}={#3}\,\ink\,{#4}}}
\newcommand{\ematch}[2]{\termc{\matchk\,{#1}\,\withk\,{#2}}}

\newcommand{\evariant}[3]{\termc{\langle{#1}={#2}\rangle\,\ask\,\typec{#3}}}
\newcommand{\ecase}[2]{\termc{\casek\,{#1}\,\ofk\,{#2}}}
\newcommand{\epair}[2]{\termc{\{{#1},{#2}\}}}

\newcommand{\Endl}{\keyword{End}}
\newcommand{\Duall}{\keyword{Dual}}

\newcommand{\FM}[1]{F^{#1}}
\newcommand{\FMOmega}[1]{F^{#1}_\omega}
\newcommand{\FMu}{\FM{\mu}}
\newcommand{\FMuDot}{\FM{\mu\cdot}}
\newcommand{\FMuSemi}{\FM{\mu;}}
\newcommand{\FMuOmega}{\FMOmega{\mu}}
\newcommand{\FMuDotOmega}{\FMOmega{\mu\cdot}}
\newcommand{\FMuSemiOmega}{\FMOmega{\mu;}}
\newcommand{\FMuAstOmega}{\FMOmega{\mu_\ast}}
\newcommand{\FMuAstDotOmega}{\FMOmega{\mu_\ast\cdot}}
\newcommand{\FMuAstSemiOmega}{\FMOmega{\mu_\ast;}}

\newcommand{\emptyCtx}{\cdot}

\newcommand{\PAR}{\mid}
\newcommand{\NU}[3]{\termc{(\nu{#1}{#2}){#3}}}
\newcommand{\thread}[1]{\termc{\langle{#1}\rangle}}

\newcommand{\subs}[2]{\blk[{\typec{#1}}\blk/{\typec{#2}}\blk]} 
\newcommand{\vsubs}[2]{\blk[\termc{#1}\blk/\termc{#2}\blk]} 
\newcommand{\nsubs}[2]{\nsymc{#1}\blk/\nsymc{#2}}

\newcommand{\judgementlabel}[1]{\mathsf{#1}} 

\newcommand{\judgement}[2]{{#1} \: \judgementlabel{#2}}

\newcommand{\judgementrel}[3]{{#1} \; {#2} \; {#3}}

\newcommand{\judgementrelctxlab}[4]{{#1} \vdash \judgementrel{#2}{#3}{#4}}

\newcommand{\istype}[3]{\judgementrelctxlab{#1}{\typec{#2}}{:}{\kindc{#3}}}
\newcommand{\prekind}[3]{{#1} \vdash_{\keyword{pre}}\judgementrel{\typec{#2}}{:}{\kindc{#3}}}

\newcommand{\isterm}[4]{{#1}\mid{#2}\vdash\termc{#3}\colon\typec{#4}}
\newcommand{\termsynth}[5]{{#1}\mid{#2}\vdash\termc{#3}\Rightarrow\typec{#4}\mid{#5}}
\newcommand{\termagainst}[5]{{#1}\mid{#2}\vdash\termc{#3}\Leftarrow\typec{#4}\mid{#5}}

\newcommand{\isproc}[2]{{#1}\vdash{\termc{#2}}}

\newcommand{\isEquiv}[3]{\judgementrel{\typec{#1}}{#2}{\typec{#3}}}

\DeclareDocumentCommand{\isbisim} { m m } { \isEquiv{#1}\sim{#2} }
\DeclareDocumentCommand{\isbisimg} { m m } { \judgementrel{\nsymc{#1}}{\gequiv}{\nsymc{#2}} }

\newcommand{\iseqt}[2]{\judgementrel{\typec{#1}}{\Eq}{\typec{#2}}} 

\newcommand{\isscong}[2]{\judgementrel{\termc{#1}}{\equiv}{\termc{#2}}}

\newcommand{\ltsarrow}[1]{\stackrel{#1}{\longrightarrow}}
\newcommand{\betaarrow}{\longrightarrow}
\newcommand{\muarrow}{\longrightarrow_\mu}
\newcommand{\lseqarrow}{\longrightarrow_{\beta;D}}
\newcommand{\gltsred}[3]{\nsymc{#1}\ltsarrow{\tsymc{#2}}\nsymc{#3}}
\newcommand{\ltsred}[3]{\typec{#1}\ltsarrow{\tsymc{#2}}\typec{#3}}

\newcommand{\betared}[2]{\typec{#1}\betaarrow\typec{#2}} 
\newcommand{\betareds}[2]{\typec{#1}\betaarrow^*\typec{#2}} 
\newcommand{\mured}[2]{\typec{#1}\muarrow\typec{#2}} 
\newcommand{\lseqred}[2]{\typec{#1}\lseqarrow\typec{#2}} 
\newcommand{\lambdared}[2]{\typec{#1}\betaarrow\typec{#2}} 
\newcommand{\normalred}[2]{\typec{#1}\Downarrow\typec{#2}} 
\newcommand{\munormalred}[2]{\typec{#1}\Downarrow_\mu\typec{#2}} 
\newcommand{\lseqnormalred}[2]{\typec{#1}\Downarrow_{\beta;D}\typec{#2}} 
\newcommand{\expred}[2]{\termc{#1}\rightarrow\termc{#2}} 
\newcommand{\procred}[2]{\expred{#1}{#2}} 

\newcommand{\iswhnf}[1]{\judgement{\typec{#1}}{whnf}} 

\newcommand{\isnormalised}[1]{\judgement{\typec{#1}}{norm}} 

\ifx\vv\undefined
\newcommand{\vv}[1]{\marginpar{\textcolor{blue}{#1}}}
\else
\renewcommand{\vv}[1]{\marginpar{\textcolor{blue}{#1}}}
\fi

\newcommand{\ie}{i.e.,\xspace} 
\newcommand{\eg}{e.g.\xspace}  
\newcommand{\st}{s.t.\xspace}  
\newcommand{\etal}{et al.\xspace} 
\newcommand{\cf}{cf.~} 

\lstdefinestyle{eclipse}{
  breaklines=true,
  basicstyle=\sffamily\Small,
  emphstyle=\color{RoyalBlue}\bfseries,
  keywordstyle=\color{RedViolet}\bfseries,
  showstringspaces=false,
}

\newcommand{\construlename}{Const}

\newcommand{\tabsrulename}{TAbs}

\newcommand{\tapprulename}{TApp}

\newcommand{\varrulename}{Var}

\newcommand{\typingrulename}[1]{T-{#1}\xspace}
\newcommand{\tconstr}{\typingrulename{Const}}
\newcommand{\tvarr}{\typingrulename{Var}}

\newcommand{\tabsr}{\typingrulename{Abs}}
\newcommand{\trecr}{\typingrulename{Rec}}
\newcommand{\tappr}{\typingrulename{App}}
\newcommand{\ttabsr}{\typingrulename{TAbs}}
\newcommand{\ttappr}{\typingrulename{TApp}}
\newcommand{\trecordr}{\typingrulename{Record}}
\newcommand{\tprojr}{\typingrulename{Proj}}

\newcommand{\tcaser}{\typingrulename{Case}}
\newcommand{\tmatchr}{\typingrulename{Match}}

\newcommand{\teqr}{\typingrulename{Eq}}
\newcommand{\tweakr}{\typingrulename{Weakening}}
\newcommand{\tcontrr}{\typingrulename{Contraction}}
\newcommand{\tderr}{\typingrulename{Dereliction}}

\newcommand{\rulename}[1]{\text{\small\sc #1}\xspace}

\newcommand{\kindrulename}[1]{\rulename{K-{#1}}}

\newcommand{\kconst}{\kindrulename{\construlename}}

\newcommand{\kvar}{\kindrulename{\varrulename}}

\newcommand{\ktabs}{\kindrulename{\tabsrulename}}

\newcommand{\ktapp}{\kindrulename{\tapprulename}}

\newcommand{\redrulename}[1]{\rulename{R-{#1}}}
\newcommand{\rmu}{\redrulename{$\mu$}}
\newcommand{\rtappl}{\redrulename{TAppL}}

\newcommand{\rseqone}{\redrulename{Seq1}}
\newcommand{\rseqtwo}{\redrulename{Seq2}}
\newcommand{\rassoc}{\redrulename{Assoc}}
\newcommand{\rbeta}{\redrulename{$\beta$}}
\newcommand{\rdin}{\redrulename{D$?$}}
\newcommand{\rdout}{\redrulename{D$!$}}

\newcommand{\rddvar}{\redrulename{DDVar}}
\newcommand{\rdextc}{\redrulename{D$\&$}}
\newcommand{\rdintc}{\redrulename{D$\oplus$}}
\newcommand{\rdsemi}{\redrulename{D$;$}}
\newcommand{\rdskip}{\redrulename{DSkip}}
\newcommand{\rdend}{\redrulename{DEnd}}
\newcommand{\rdctx}{\redrulename{DCtx}}

\newcommand{\whnfrulename}[1]{\rulename{W-{#1}}}
\newcommand{\wconszero}{\whnfrulename{Const0}}
\newcommand{\wconsone}{\whnfrulename{Const1}}
\newcommand{\wvar}{\whnfrulename{Var}}
\newcommand{\wabs}{\whnfrulename{Abs}}

\newcommand{\wseqone}{\whnfrulename{Seq1}}
\newcommand{\wseqtwo}{\whnfrulename{Seq2}}
\newcommand{\wdual}{\whnfrulename{Dual}}
\newcommand{\normalisationrulename}[1]{\rulename{N-{#1}}}

\newcommand{\premspace}{\quad\;}

\begin{document}

\title{\texorpdfstring{System $F^\mu_\omega$ with Context-free Session Types%
\thanks{Support for this research was provided by the Fundação para a Ciência e a
Tecnologia through project SafeSessions, ref.\ PTDC/CCI-COM/6453/2020, and by the
LASIGE Research Unit, ref.\ UIDB/00408/2020 and ref.\ UIDP/00408/2020.
}}
{System F-mu-omega with Context-free Session Types}
}

\author{Diana Costa\orcidID{0000-0002-8312-429X} \and
Andreia Mordido\orcidID{0000-0002-1547-0692} \and
Diogo Poças\orcidID{0000-0002-5474-3614} \and
Vasco T. Vasconcelos\orcidID{0000-0002-9539-8861}}
\authorrunning{D. Costa et al.}
\institute{LASIGE, Faculdade de Ciências, Universidade de Lisboa, Portugal
\email{\{dfdcosta,afmordido,dmpocas,vmvasconcelos\}@ciencias.ulisboa.pt}}

\maketitle

\begin{abstract}
  We study increasingly expressive type systems, from $F^\mu$---an extension of
  the polymorphic lambda calculus with equirecursive types---to
  $F^{\mu;}_\omega$---the higher-order polymorphic lambda calculus with
  equirecursive types and context-free session types.
  Type equivalence is given by a standard bisimulation defined over a novel
  labelled transition system for types.
  Our system subsumes the contractive fragment of $F^\mu_\omega$ as studied in
  the literature.
  %
  Decidability results for type equivalence of the various type languages are
  obtained from the translation of types into objects of an appropriate
  computational model: finite-state automata, simple grammars and deterministic
  pushdown automata.
  We show that type equivalence is decidable for a significant fragment of the
  type language.
  We further propose a message-passing, concurrent functional language equipped
  with the expressive type language and show that it enjoys preservation and
  absence of runtime errors for typable processes.

\end{abstract}



\section{Introduction}

Extensions of the $\lambda$-calculus to include increasingly sophisticated type
structures have been extensively studied and have led to systems whose
importance is widely recognized: System $F$~\cite{DBLP:conf/programm/Reynolds74}, System $F^{\mu}$~\cite{DBLP:conf/icfp/GauthierP04}, System
$F_{\omega}$~\cite{girard1972interpretation}, System $\FMuOmega$~\cite{DBLP:conf/popl/CaiGO16}. Ideally, we would like to combine a
\emph{wishlist} of type structures and get a super-powerful system with vast
expressiveness. However, the expressiveness of types is naturally limited by the
universe where they are supposed to live: programming languages. Expressive type
systems pose challenges to compilers that other (less expressive) types do not
even reveal; one such example is type equivalence checking.

System $F$ can be enriched with different type constructors for specifying
communication protocols. We analyse the impact of combinations of such
constructors on the type equivalence problem.
In order to do so, we extend System $F$ with session
types~\cite{DBLP:conf/concur/Honda93,DBLP:conf/esop/HondaVK98,DBLP:conf/parle/TakeuchiHK94}.
Session types
provide for detailed protocol specifications in the form of types.
%
%
%
Traditional recursive session types are limited to tail recursion, thus failing
to capture all protocols whose traces cannot be characterized by regular
languages. Context-free session types overcome this limitation by extending
types with a notion of sequential composition, $\semit T U$
\cite{Almeida22,DBLP:conf/icfp/ThiemannV16}. The set of types together with the
binary operation $\typec ;$ constitutes a monoid, for which a new type, $\Skip$,
acts as the neutral element and $\End$ acts as an absorbing element.

The regular recursive type
$\tmuinfix{\alpha}{\skind}{\extchoice\records{\donel\colon\End,\morel\colon?\Int;\alpha}}$
describes an integer \emph{stream} as seen from the point of view of the
consumer. It offers a choice between $\donel$---after which the channel must
be closed, as witnessed by type $\End$---and $\morel$---after which an integer
value must be received, followed by the rest of the stream.  Types are
categorised by \emph{kinds}, so that we know that the recursion variable
$\typec\alpha$ is of kind session---denoted by $\skind$---and, thus, can be used
with semicolon.
Instead, we might want to write a type with a more \emph{context-free} flavour.
The type
$\tmuinfix{\alpha}{\skind}{\&\{\leafl: \Skip, \nodel:
  \alpha;?\Int;\alpha\}}\typec;\End$ describes a protocol for the type-safe
streaming of integer \emph{trees} on channels. The continuation to the $\leafl$
option is $\Skip$, where no communication occurs but the channel is still open
for further composition. The continuation to the $\nodel$ choice receives a left
subtree, an integer at the root and a right subtree. In either case, once the
whole tree is received, the channel must be closed, as witnessed by the final
$\End$.
Beyond first-order context-free session types (where only basic types are
exchanged)~\cite{Almeida22,DBLP:conf/icfp/ThiemannV16} we may be interested in
higher-order session types capable of exchanging values of complex
types~\cite{DBLP:journals/corr/abs-2203-12877}.
%
%
%
A goal of this paper is the integration of higher-order context-free session
types into system $\FMuOmega$. We want to be able to abstract the type that
is received on a tree channel, which is now possible by writing $\tabs \alpha \tkind
{\tmuinfix{\beta}{\skind}{\&\{\leafl: \Skip, \nodel: \beta;?\alpha;\beta\}}}\typec;\End$,
where $\tkind$ is the kind of functional types.

A form of abstraction over session types was formerly proposed by Das
\etal~\cite{DBLP:conf/esop/DasDMP21,DBLP:journals/toplas/DasDMP22} via (nested)
parametric polymorphism. 
In the notation of Das \etal, we can write a type equation 
$\iseqt{\callt\streamt\alpha}{\extchoice\records{\donel\colon\End,\morel\colon?\alpha;\callt\streamt\alpha}}$ 
that abstracts the type being received on a stream channel. 
To write the same type using abstraction, we can think of $\streamt$ as a function of its parameter $\typec\alpha$, 
$\iseqt{\streamt}{\tabs\alpha\tkind{\extchoice\records{\donel\colon\End,\morel\colon?\alpha;\tapp\streamt\alpha}}}$; 
we can then rewrite $\streamt$ using the $\typec\mu$-operator,
$\streamt = \tabs\alpha\tkind{(\tmuinfix{\beta}{\skind}{\extchoice\records{\donel\colon\End,\morel\colon\IN\alpha\beta}})}.$
Das \etal proved that parametrized type definitions
over regular session types are strictly more expressive than context-free
session types. To some extent, this analogy guides our approach: if adding
abstraction (via parametric polymorphism) to regular types leads to nested types, what exactly does it
mean to add abstraction (via a type-level $\lambda$-operator) to context-free types? Throughout this paper we analyse
several increments to System $F^\mu$ that culminate in adding
$\lambda$-abstraction to context-free session types.

One of our focuses is necessarily the analysis of the type equivalence problem.
The uncertainty about the decidability of this problem over recursive parametric
types goes back to the
1970s~\cite{DBLP:journals/csur/CardelliW85,DBLP:conf/popl/Solomon78}. Although the
type equivalence problem for parametric (nested) session types and context-free
session types is decidable, that for the combination of abstractions over
context-free types may no longer be. In fact, this analysis constitutes an
interesting journey towards a better understanding of the role of higher-order
polymorphic recursion in presence of sequential composition, as well as the
gains (and losses) resulting from combining abstraction with arbitrary (rather
than tail) recursion.

Ultimately, decidability is not a sufficiently valuable measure regarding a type
system's \emph{practicality}. We look for type systems that may be incorporated
into compilers. For that reason, we are interested in algorithms for type
equivalence checking. Equivalence in $\FMuOmega$ alone is already at least as
expressive as deterministic pushdown automata. If we restrict recursion to the
monomorphic case (requiring recursion variables to denote proper types, that is
of kind $\skind$ or $\tkind$, collectively denoted by $\kast$)
we lower the complexity of type equivalence to that of equivalence for
finite-state automata. The extension with context-free session types is slightly
more complex. In order to obtain ``good'' algorithms, we restrict the recursion
to the monomorphic case,
arriving at classes $\FMuAstOmega, \FMuAstSemiOmega$. Now the type equality
problem for $\FMuAstSemiOmega$ translates to the equivalence problem for simple
grammars, which is still decidable \cite{DBLP:conf/tacas/AlmeidaMV20,DBLP:conf/fossacs/GayPV22}.
Since $\FMuAstSemiOmega$ subsumes $\FMuAstOmega$, our proof of the
decidability of type equivalence serves as an alternative to that of Cai
\etal~\cite{DBLP:conf/popl/CaiGO16} (restricted to contractive types).

Higher-order polymorphism allows for the definition of type operators and the
internalisation of various (session-type) constructs that would otherwise be
offered as built-in constructors. In this way, we are able to internalise basic
session-type constructors such as sequential composition $\typec;$ and the
$\Dual$ type operator (which reverses the direction of communication between
parties).
Duality is often treated as an external macro. Gay
\etal\cite{DBLP:journals/corr/abs-2004-01322} explore different ways of handling
the dual operator, all in a monomorphic setting. In the presence of polymorphism
the dual operator cannot be fully eliminated without introducing co-variables.
Internalisation offers a much cleaner solution.



Due to the presence of sequential composition, regular trees are \emph{not} a
powerful enough model for representing types (\lstinline|type TreeC a| in
\cref{sec:motivation} is an example).
The main technical challenge
when combining System $\FMuOmega$ and context-free session types is making sure
that the resulting model can still be represented by simple grammars, so that
type equivalence may be decided by a practical algorithm.
The difficulties arise with renaming bound variables.
For infinite types, both renaming with fresh variables and using de Bruijn
indices may create an infinite number of distinct variables, which makes the
construction of a simple grammar simply impossible. For example, take the type
$\tabs\alpha\tkind{\tmuinfix\gamma\tkind{\tabs\beta\tkind{\function\alpha\gamma}}}$,
which stands for the infinite type
$\tabs\alpha\tkind{\tabs\beta\tkind{\function\alpha{\tabs\beta\tkind{\function\alpha{\tabs\beta\tkind..}}}}}$
Renaming this type using a fresh variable at each step would result in a type of
the form
$\tabs{\upsilon_1}\tkind{\tabs{\upsilon_2}\tkind{\function{\upsilon_1}{\tabs{\upsilon_3}\tkind{\function{\upsilon_1}{\tabs{\upsilon_4}\tkind..}}}}}$,
requiring infinitely many variables. Similarly, de Bruijn
indices~\cite{debruijn:1972:lambda} yield a type of the form
$\typec{\lambda_\tkind\lambda_\tkind\function 1{\lambda_\tkind\function
    2{\lambda_\tkind\function 3\ldots}}}$ that requires an infinite number of
natural indices.
We thus introduce \emph{minimal renaming} that uses the least amount of variable
names as possible (\cf Gauthier and Pottier~\cite{DBLP:conf/icfp/GauthierP04}).
This ensures that only finitely many terminal symbols are necessary, allowing
for translating types into simple grammars.




Type languages live in term languages and we propose a term language to consume
$\FMuSemiOmega$ types. Based on Almeida
\etal~\cite{DBLP:journals/corr/abs-2106-06658}, we introduce a message-passing
concurrent programming language. Type checking is decidable if type equivalence
is, and it is, in particular, for $\FMuAstSemiOmega$. 

The main contributions of this paper are as follows.
\begin{itemize}
\item The integration of (higher-order) context-free session types into system
  $\FMuOmega$, dubbed $\FMuSemiOmega$.
\item A semantic definition of type equivalence via a labelled transition system.
\item The identification of a suitable fragment of System $\FMuSemiOmega$ for which type equivalence is reduced to the bisimilarity of simple grammars.
\item A proof that type equivalence on the full System $\FMuSemiOmega$ is at least as hard as bisimilarity of deterministic pushdown automata.
\item The first internalisation of the $\Dual$ type operator in a type language.
\item A term language to consume $\FMuSemiOmega$ types and an accompanying metatheory.
\end{itemize}

The type system presented in the paper combines three constructions: sequential composition of session types, higher-order kinds via type-level abstraction and application, and higher-order recursion. 
Prior to our work there is the system by Almeida \etal\cite{DBLP:conf/tacas/AlmeidaMV20} which incorporates sequential composition and (first-order) recursion, but no higher-order kinds. 
There is also the system by Cai \etal~\cite{DBLP:conf/popl/CaiGO16} which incorporates higher-order kinds and higher-order recursion, but no sequential composition. 
Our system is the first to incorporates all three constructions. 
Although some of the results are incremental and generalize results from the literature, the main technical challenge is understanding the border past which they don't hold anymore. 
For example, ``just'' including higher-order kinds into the system by Almeida \etal does not work, since we need to pay close attention to variable names, making sure that type equivalence is invariant with respect to alpha-conversion (renaming of bound variables). 
This required us to define a novel notion of renaming, inspired by Gauthier and Pottier~\cite{DBLP:conf/icfp/GauthierP04}. 
Similarly, ``just'' including sequential composition into the system of Cai \etal does not work, since finite-state automata (or regular trees) are not enough to capture the expressive power of the new type system, \emph{even} when restricted to first-order recursion. 
This required us to look at the more expressive framework of simple grammars, and introduce a translation from types to words of a simple grammar.

The rest of the paper is organised as follows. The next section motivates the
type language and introduces the term language with an example. \Cref{sec:types}
introduces System $\FMuSemiOmega$, \cref{sec:type-equiv} discusses type
equivalence and \cref{sec:equivalence-decidable} shows that type equivalence is
decidable for a fragment of the type language. \Cref{sec:processes}
presents the term language and its metatheory. \Cref{sec:related} discusses
related work and \cref{sec:conclusion} concludes the paper with pointers for
future work. 



\section{Motivation}
\label{sec:motivation}


\begin{figure}[t!]
  \begin{align*}
    \typec{\varrecs{}} \grmeq{} \typec{\records{}} \grmor \typec{\variants{}}
    &&
    \typec\sharp \grmeq{} \typec? \grmor{} \typec!
    &&
    \choice \grmeq{} \extchoice \grmor \intchoice
  \end{align*}
  \begin{align*}
    \typec T \grmeq & \function TT
    \grmor \tbananas{l_i}{T_i} 
    \grmor  \foralltinfix\alpha\kind T 
    \grmor \tmuinfix{\alpha}{\kind}{T} 
    \grmor \typec{\alpha}
    && (\FMu)
    \\
    \typec T \grmeq & (\FMu)
    \grmor \MSG TT
    \grmor \tchoice{l_i}{T_i}
    \grmor \End && (\FMuDot)
    \\
    \typec T \grmeq & (\FMu)
    \grmor{} \MSGn T
    \grmor{} \tchoice{l_i}{T_i}
    \grmor \End
    \grmor{} \semit TT 
    \grmor \Skip
    && (\FMuSemi)
    \\
    \typec T \grmeq & (\FM{M})
    \grmor \tabs\alpha\kind T
    \grmor \tapp TT && (\FMOmega{M}),\quad M \grmeq \mu\ , \mu\cdot\ , \mu;
  \end{align*}
  \caption{Six $F$-systems.}
  \label{fig:F-systems}
\end{figure}

Our goal is to study type systems that combine equirecursion, higher-order
polymorphism, and higher-order context-free session types and their
incorporation in programming languages.

\paragraph{Extensions of $F$.}
\Cref{fig:F-systems} motivates the construction by proposing six different type
syntaxes, culminating with $\FMuSemiOmega$.
%
The initial system, $\FMu$, includes well-known basic type
operators~\cite{DBLP:books/daglib/0005958}: functions $\function TU$, records
$\trecord{l_i}{T_i}$ and variants $\tvariant{l_i}{T_i}$. Type $\Unit$ is short
for $\typec{\{\}}$, the empty record; we can imagine that $\Unit$ stands in place of an arbitrary scalar type such as $\Int$ and
$\Bool$. 
We also include variable names $\typec\alpha$, type quantification
$\foralltinfix\alpha\kind T$ and recursion $\tmuinfix{\alpha}{\kind}{T}$.
%
%
In order to control type formation, all variable bindings must be kinded with
some kind $\kind$, even if for the initial system, $\FMu$, we only use the
functional kind $\tkind$.

We then build on $\FMu$ by considering (regular, tail recursive) session types;
we represent the resulting system by $\FMuDot$.
For example $\IN\Int{\OUT\Bool\End}$ is a type for a channel endpoint that
receives an integer, sends a boolean, and terminates.
At this point we introduce a kind $\skind$ of session types to
restrict the ways in which we can combine session and functional types together.
For example, a well-formed type $\IN TU$ is of kind $\skind$ and requires $\UT$
to be also of kind $\skind$ (whereas $\TT$ can be of kind $\kast$, that is $\skind$
or $\tkind$).
An example of an infinite session type is
$\tmuinfix{\alpha}{\skind}{\OUT\Int\alpha}$ that endlessly outputs integer
values. For a more elaborate example consider the type 
$\intstreamt = \tmuinfix{\alpha}{\skind}{\extchoice\records{\donel\colon\End,\morel\colon\IN\Int\alpha}}$
that specifies a channel endpoint for receiving a (finite or infinite) stream of
integer values.
Communication ends after choice $\donel$ is selected.

The next step of our construction takes us to context-free session types; the
resulting system is denoted by $\FMuSemi$. We introduce a new construct for
sequential composition $\semit TU$, and a new type $\Skip$, acting as the
neutral element of sequential composition~\cite{DBLP:conf/icfp/ThiemannV16}. The
message constructors are now unary ($\INn T$ and $\OUTn T$) rather than binary.
In System $F^{\mu;}$ we distinguish between the traditional $\End$ type and the
$\Skip$ type. These types have different behaviours: $\End$ terminates a
channel, while $\Skip$ allows for further communication. Type equality is more
subtle for context-free session types, because of the monoidal semantics of
sequential composition. It is derivable from the following axioms:
\begin{equation}\label{eq:axiomssemi}\begin{aligned}
  \typec{\semit\Skip T} &\sim \typec{T} && \text{Neutral element}
  \\
  \typec{\semit\End T} &\sim \typec{\End} && \text{Absorbing element}
  \\
  \typec{\semit{(\semit TU)}{V}} &\sim {\semit{T}{(\semit UV)}} && \text{Associativity}
  \\
  \typec{\semit{\tchoice{l_i}{T_i}} U} &\sim {\tchoice{l_i}{\semit{T_i}U}}
   && \text{Distributivity}
\end{aligned}\end{equation}


Although the syntax of $\FMuDot$ is not formally included in the syntax of
$\FMuSemi$, we can embed recursive session types into context-free session types
by mapping $\MSG TU$ into $\semit{\MSGn T} U$. It is well-known that
context-free session types allow for higher computational expressivity: while
$\FMu$ and $\FMuDot$ can be represented via finite-state automata, $\FMuSemi$
can only be represented with simple grammars
\cite{DBLP:conf/tacas/AlmeidaMV20,DBLP:conf/fossacs/GayPV22}.

To finalise our construction, we include type abstraction $\tabs\alpha\kind T$
and type application $\TT\ \UT$. Again, type abstraction binds a variable which
must be kinded. Kinds can now be of higher-order $\karrow{\kind}{\kind'}$. For
each of the three systems $\FMu$, $\FMuDot$, $\FMuSemi$ we arrive at a
higher-order version, respectively $\FMuOmega$, $\FMuDotOmega$, $\FMuSemiOmega$
(all of which we represent as $\FMOmega{M}$). In System $\FMuDotOmega$, for
example, we can specify channels for receiving (finite or infinite) sequences of
values of arbitrary (but fixed) types,
\begin{equation*}
\streamt = \tabs\alpha\tkind{(\tmuinfix{\beta}{\skind}{\extchoice\records{\donel\colon\End,\morel\colon\IN\alpha\beta}})}
\end{equation*}
where $\typec \alpha$ can be instantiated with the desired type; in particular,
$\tapp \streamt \Int$ would be equivalent to the aforementioned $\intstreamt$.

%

It turns out that the expressive power of general higher-order systems
$\FMOmega{M}$ is too large for practical purposes. Even the simplest case
$\FMuOmega$ is at least as expressive as deterministic pushdown automata (or
equivalently, first-order grammars), for which known equivalence algorithms are
notoriously impractical. By impractical we mean that, although there exists a
proof of decidability (due to Sénizergues~\cite{DBLP:conf/icalp/Senizergues97},
later improved by Stirling and
Jancar~\cite{DBLP:journals/corr/abs-1010-4760,DBLP:conf/icalp/Stirling02}), the
underlying algorithm is rather complex. To the best of our knowledge, there
is no practical implementation of an algorithm to decide the equivalence of
deterministic pushdown automata. This is essentially due to polymorphic
recursion, which can be encoded by a higher-order $\mu$-operator (we provide an example at the end of \cref{sec:equivalence-decidable}). Therefore, it makes sense to restrict the kind $\kind$ of the
recursion operator $\tmuinfix\alpha\kind T$. We use the notation $\mu_\ast$ to
mean the subclass of types written using only $\kast$-kinded recursion, \ie
$\tmuinfix\alpha\tkind T$ or $\tmuinfix\alpha\skind T$.

\begin{figure}[t!]
  \centering
\begin{tikzpicture}[thick, node distance = 1.2cm, every node/.style = {font = \normalsize}]
\node (fmu) {$\FMu$};
\node[below of = fmu] (fmudot) {$\FMuDot$};
\node[below of = fmudot] (fmusemi) {$\FMuSemi$};
\node[right of = fmu, xshift = 0.4cm] (fmuomega0) {$\FMuAstOmega$};
\node[below of = fmuomega0] (fmudotomega0) {$\FMuAstDotOmega$};
\node[below of = fmudotomega0] (fmusemiomega0) {$\FMuAstSemiOmega$};
\node[right of = fmuomega0, xshift = 0.4cm] (fmuomega) {$\FMuOmega$};
\node[below of = fmuomega] (fmudotomega) {$\FMuDotOmega$};
\node[below of = fmudotomega] (fmusemiomega) {$\FMuSemiOmega$};
\node[above right of = fmuomega0, xshift = -0.1cm, yshift = -0.4cm] (blue0) {};
\node[below of = blue0] (blue1) {};
\node[below of = blue1, yshift = .25cm] (blue2) {};
\node[left of = blue2, xshift = -0.3cm] (blue3) {};
\node[left of = blue3, xshift = -1cm] (blue4) {};
\draw[blue, line width = 0.1cm, rounded corners] (blue0.center) -- (blue1.center) -- (blue2.center) -- (blue3.center) -- (blue4.center);
\node[blue, above of = blue4, yshift = -0.7cm, align = right, font = \small] (bluel) {\textbf{finite-state}\\\textbf{automata}};
\node[below right of = fmusemiomega0, xshift = -0.1cm, yshift = 0.1cm] (red0) {};
\node[above of = red0, yshift = 0.1cm] (red1) {};
\node[left of = red1, xshift = -0.3cm] (red2) {};
\node[left of = red2, xshift = -1cm] (red3) {};
\draw[red, line width = 0.1cm, rounded corners] (red0.center) -- (red1.center) -- (red2.center) -- (red3.center);
\node[red, below of = red3, xshift = .2cm, yshift = 0.7cm, align = right, font = \small] (redl) {\textbf{simple}\\\textbf{grammars}};
\node[above right of = fmuomega0, xshift = 0.05cm, yshift = -0.4cm] (green0) {};
\node[below of = green0] (green1) {};
\node[below of = green1] (green2) {};
\node[below of = green2, yshift = .1cm] (green3) {};
\node[right of = green3, xshift = .7cm] (green4) {};
\node[OliveGreen, above right of = green4, yshift = -0.4cm, align = left, font = \small] (greenl) {$\geq$ \textbf{deterministic}\\\textbf{pushdown} \textbf{automata}};
\draw[OliveGreen, line width = 0.1cm, rounded corners] (green0.center) -- (green1.center) -- (green2.center) -- (green3.center) -- (green4.center);
\draw[->] (fmu) edge (fmudot);
\draw[->] (fmudot) edge (fmusemi);
\draw[->] (fmu) edge (fmuomega0);
\draw[->] (fmudot) edge (fmudotomega0);
\draw[->] (fmusemi) edge (fmusemiomega0);
\draw[->] (fmuomega0) edge (fmudotomega0);
\draw[->] (fmudotomega0) edge (fmusemiomega0);
\draw[->] (fmuomega0) edge (fmuomega);
\draw[->] (fmudotomega0) edge (fmudotomega);
\draw[->] (fmusemiomega0) edge (fmusemiomega);
\draw[->] (fmuomega) edge (fmudotomega);
\draw[->] (fmudotomega) edge (fmusemiomega);
\end{tikzpicture}
\caption{Relation between the main classes of types in this paper (arrows denote strict inclusions).}
\label{fig:fmudiagram}
\end{figure}
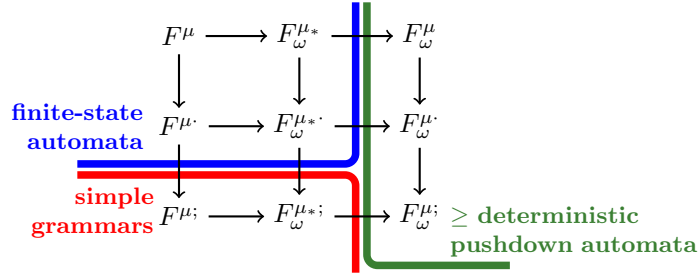


\Cref{fig:fmudiagram} summarizes the main relations between the classes of types
in our paper. Firstly, we obtain a lattice where the expressive power increases
as we travel down (from functional to session to context-free session types) and
right (from simple polymorphism to higher-order polymorphism with monomorphic recursion
to arbitrary recursion). Four of the classes can be represented
using finite-state automata (up to $\FMuAstDotOmega$). By including sequential
composition 
($\FMuSemi$ and $\FMuAstSemiOmega$) we are still able to represent types using
simple grammars. Once we allow for arbitrary recursion, the expressiveness of
our model requires the computational power of deterministic pushdown automata.


\paragraph{Programming with $\FMuSemiOmega$.}

We now turn our attention to the term language, a message passing, concurrent
functional language, equipped with context-free session types.
Start with a stream of values of type \lstinline|a|. Such a stream, when seen
from the side of the reader, offers two choices: \lstinline|Done| and
\lstinline|More|. In the former case the interaction is over; in the latter the
reader reads a value of type \lstinline|a|, as in \lstinline|?a|, and recurses.
This is the stream type we have seen before only that, rather than closing the
channel endpoint (with type \lstinline|End|), it terminates with type
\lstinline|Skip|, so that it may be sequentially composed with other types. In
this informal introduction to the term language we omit the kinds of type
variables.
\begin{lstlisting}
type Stream a = &{Done: Skip, More: ?a ; Stream a}
\end{lstlisting}

A fold channel, as seen from the side of the folder, is a type of the following
form. We assume that application binds tighter than semicolon, that is, type
\lstinline|Stream a ; !b ; End| is interpreted as \lstinline|(Stream a)  ; !b ; End|.
\begin{lstlisting}
type Fold a b = ?(b -> a -> b) ; ?b ; Stream a ; !b ; End
\end{lstlisting}
Consumers of this type first receive the folding function, then the starting
element, then the elements to fold in the form of a stream, and finally output
the result of the fold. The type terminates with \lstinline|End| for we do not
expect type \lstinline|Fold| to be further composed. Compare \lstinline|Fold|
with the type for a conventional functional left fold:
\lstinline|(b -> a -> b) -> b -> List a -> b|.

We now develop a function that consumes a \lstinline|Fold| channel. Syntax
\lstinline@x |> f@ is for the inverse function application with low priority,
that is \lstinline@x |> f |> g = g (f x)@. Recall that \lstinline|Unit| is an
alternative notation for the empty record type, \lstinline|{}|.
\begin{lstlisting}
foldServer : foralla.forallb. Fold a b -> Unit
foldServer c = let (f, c) = receive c in
               let (e, c) = receive c in foldS f e c

foldS : foralla.forallb. (b -> a -> b) -> b -> Stream a;!b;End -> Unit
foldS f e c = match c with
  { Done c -> c |> send e |> close
  , More c -> let (x, c) = receive c in foldS f (f e x) c
  }
\end{lstlisting}
Function \lstinline|foldServer| consumes the initial part of the channel and
passes the rest of the channel to the recursive function \lstinline|foldS| that
consumes the whole stream while accumulating the fold value. In the end, when
branch \lstinline|Done| is selected, the fold value is written on the channel
and the channel closed.
In general, the channel operators---\lstinline|receive|, \lstinline|send|,
\lstinline|select|---return the same channel in the form of a new identifier. It
is customary to reuse the identifier name---\lstinline|c| in the example, as in
\lstinline|let (f, c) = receive c|---since it denotes the same channel. Syntax
\lstinline@c|>...@ hides the continuation channel. The case for the external
choice---\lstinline|match|---also returns the continuation (in each branch) so
that interaction on the channel endpoint may proceed.

We may now write different clients for the \lstinline|foldServer|. Examples
include a client that generates a stream from a pair of integer values (denoting
an interval); another that generates the stream from a list of values; and yet
another that generates the stream from a binary tree.
We propose a further client. Consider the type of a channel that exchanges trees
in a serialized format \cite{DBLP:conf/icfp/ThiemannV16}. Its polymorphic
version, as seen from the point of view of the reader, is as follows:
\begin{lstlisting}
type TreeChannel a = TreeC a ; End
type TreeC a = &{Leaf: Skip, Node: TreeC a;?a;TreeC a}
\end{lstlisting}

We transform trees as we read from tree channels into streams. Function
\lstinline|flatten| receives a tree channel and a stream channel (as seen from
the point of view of the writer, hence the \lstinline|Dual|) and returns the
unused part of the stream channel.
\begin{lstlisting}
flatten : foralla.forallc. TreeChannel a -> (Dual Stream a);c -> c
\end{lstlisting}

We are now in a position to write a client that checks whether all values in a
tree channel are positive.
\begin{lstlisting}
allPositive : TreeChannel Int -> Dual (Fold Int Bool) -> Bool
allPositive t c =
  let c = send (lambdax:Bool.lambday:Int. x && y > 0) c in
  let c = send True c in
  let c = flatten [Int] [?Bool;End] t c in
  let (x, c) = receive c in
  close c; x
\end{lstlisting}
The client sends a function and the starting value on the fold channel. Then, it
flattens the given tree \lstinline|t|, receives the folded value and closes the
channel.
Syntax \lstinline|flatten [Int] [?Bool;End]| is for term-level type application.
We mean to flatten a tree of \lstinline|Int| values on a stream channel whose
continuation is of type \lstinline|?Bool;End|. The continuation channel is bound
to \lstinline|c| so that we may further receive the fold value and thereupon
close the channel.
Syntax \lstinline|e1;e2| is for sequential composition and abbreviates
\lstinline|(lambdax:Unit.let {} = x in e2) e1| given that \lstinline|{}|, the
\lstinline|Unit| value, is linear and hence must be consumed.

Finally, a simple application creates a new \lstinline"TreeC" channel, passing
one end to a thread that produces a tree channel. It then creates a
\lstinline|Fold| channel, distributes one end to a thread \lstinline|foldServer|
and the other to function \lstinline|allPositive|. The \lstinline|fork|
primitive receives a suspended computation (a thunk, of the form
\lstinline|lambdax:Unit.e|) and creates a new thread that runs in parallel with
that from where the \lstinline|fork| was issued.
\begin{lstlisting}
system : Bool
system = let (tr, tw) = new TreeC Int in
  fork (lambda_:Unit. produce tw);
  let (fr, fw) = new Fold Int Bool in
  fork (lambda_:Unit. foldServer fr);
  allPositive tr fw
\end{lstlisting}





\section{Kinds and Types}
\label{sec:types}


\begin{figure}[!t]
	\begin{minipage}{.45\textwidth}
		\centering
		\begin{align*}
			\kast \grmeq& && \text{\small Kind of proper types}\\
			& \skind && \text{\small session}\\
			& \tkind && \text{\small functional}\\
			\kind \grmeq& &&\text{\small Kind}\\
			& \kast && \text{\small kind of proper types}\\
			& \karrow{\kind}{\kind} && \text{\small kind of type operators}\\
			\TT \grmeq& &&\text{\small Type}\\
			& \tiota && \text{\small type constant}\\
			& \talpha && \text{\small type variable}\\
			& \tabs \alpha\kind T && \text{\small type-level abstraction}\\
			& \tapp TT && \text{\small type-level application}
		\end{align*}
		\caption{The syntax of types.}
		\label{fig:syntax-types}
	\end{minipage}
	\begin{minipage}{0.45\textwidth}
		\centering
		  \begin{align*}
			\tiota \grmeq& &&&& \text{\small{Type constant}}\\
			&\Arrow && \karrow \kast {\karrow {\kast}{\tkind}} && \text{\small arrow}\\
			&\typec{\varrecs{\overline{l_i}}} &&\kindc{\overline{{\kast \Rightarrow}}\,\tkind} && \text{\small record, variant}\\
			&\tmu \kind && \karrow {(\karrow \kind \kind)} {\kind} &&
                                                                      \text{\small recursive type}\\
			&\tforall\kind && \karrow {(\karrow \kind \kast)} \tkind &&
                                                                        \text{\small universal type}\\
			&\Skip && \skind && \text{\small{skip}}\\
			&\End && \skind && \text{\small end}\\
			&\Sharp &&\karrow \kast \skind && \text{\small input, output}\\
			&\Semi && \karrow \skind {\karrow \skind \skind} && \text{\small seq.\ composition}\\
			&\typec{\odot{\{\overline{l_i}\}}} && \kindc{\overline{{\skind
						\Rightarrow}}\,\skind} && \text{\small choice operators}\\
			&\Dual && \karrow{\skind}{\skind} && \text{\small dual operator}
		\end{align*}
		\caption{Type constants and kinds.}
		\label{fig:constants}
	\end{minipage}
\end{figure}



This section introduces in detail System $\FMuSemiOmega$, an extension of
System $\FMuOmega$ incorporating higher-order context-free session types. 
The syntax of types is presented in~\cref{fig:syntax-types}.
A type is either a constant~$\typec{\iota}$
(as in \cref{fig:constants}), a type variable $\typec{\alpha}$, an abstraction
$\tlambda{\alpha}{\kind}{\TT}$ or an application $\tapp TU$. 
Besides incorporating the standard session type constructors as constants, system
$\FMuSemiOmega$ also includes $\Dual$ as a constant for a type operator mapping a
session type to its dual. Note also that $\foralltinfix{\alpha}{\kind}{T}$ is
syntactic sugar for $\typec{\forall_\kind (\tlambda{\alpha}{\kind}{T})}$.
Analogously, $\tmuinfix{\alpha}{\kind}{T}$ abbreviates
$\typec{\tmu\kind (\tlambda{\alpha}{\kind}{T})}$. This simplifies our analysis
as lambda abstraction becomes the only binding operator. 

A distinction between session and functional types is made resorting to
kinds $\skind$ and $\tkind$, respectively. These are the kinds of proper types,
$\kast$; we use the symbol $\kind$ to represent either the kind of a proper
type or that of a type operator, of the form $\karrow {\kind} {\kind'}$. A
kinding context $\Delta$ stores kinds for type variables using bindings of the
form $\tbind \alpha \kind$. Notation $\Delta+\tbind \alpha \kind$ denotes the
update of kinding context $\Delta$, defined as
$(\Delta,\tbind \alpha \kind)+\tbind \alpha {\kind'} = \Delta,\tbind
\alpha{\kind'}$ and $\Delta+\tbind \alpha \kind = \Delta,\tbind \alpha \kind$
when $\typec\alpha\not\in\Delta$.

\begin{figure}[t!]
  \declrel{Type renaming}{$\rename_S(\TT)$}
  \begin{align*}
    \rename_S(\typec\iota)&=\typec\iota
    \\
    \rename_S(\typec\alpha)&=\typec\alpha
    \\
    \rename_S(\tabs{\alpha}{\kind}{T}) &=
    \tabs{\upsilon}{\kind}{{\blk{\rename_S(}}T\subs\upsilon\alpha})
    \quad\text{where } \typec{\upsilon} = \first_S(\tabs{\alpha}{\kind}{T})
    \\
    \rename_S(\tapp{T}{U}) &= \rename_{S\cup\fv(\UT)}(\TT)\rename_S(\UT)
  \end{align*}
  \caption{Type renaming.}
  \label{fig:rename}
\end{figure}

To define type formation, we require a few notions. Firstly comes the
notion of \emph{renaming}, adapted from Gauthier and Pottier~\cite{DBLP:conf/icfp/GauthierP04} and presented in \cref{fig:rename}. 
Renaming essentially replaces a type $\TT$ by a minimal alpha-conversion of $\TT$. By alpha-conversion we mean that $\rename_S(\TT)$ renames bound variables in $\TT$. By ``minimal'' we mean that each bound variable is renamed to its lowest possible value. 
We assume at our disposal a countable well-ordered set of type variables
$\{\typec{\upsilon_1},\ldots,\typec{\upsilon_n},\ldots\}$.
In $\rename_S(\TT)$, parameter $S$ is a set containing type variables
unavailable for renaming; in the outset of the renaming process $S$ is the empty
set, since all variables are available. In that case the subscript $S$ is often
omitted. 
The case for lambda abstraction renames the bound variable by the smallest variable not in the set
$S \cup \fv(\tabs{\alpha}{\kind}{T})$, which we denote by $\first_S(\tabs{\alpha}{\kind}{T})$.


Renaming is what allows us to check whether type abstractions $\tabs\alpha\kind T$, $\tabs\beta\kind U$ are equivalent. The types are equivalent if both bound variables $\typec\alpha$ and $\typec\beta$ are renamed to the same variable
$\typec{\upsilon_j}$. In summary, renaming provides a syntax-guided approach to the equivalence of lambda-abstractions, where the names of bound variables should not matter. Our notion of type equivalence preserves alpha-conversions up to renaming: if $\TT$, $\UT$ are alpha-conversions of one another, then $\rename(\TT) = \rename(\UT)$ and in particular $\isbisim{\renameb\TT}{\renameb\UT}$. We will come back to this point after we define type equivalence in \cref{sec:type-equiv}.

We can easily see that renaming uses the minimum amount of variable names
possible; for example,
$\rename(\tlambda{\alpha}{\tkind}{\tlambda{\beta}{\skind}{\beta}}) =
\tlambda{\upsilon_1}{\tkind}{\tlambda{\upsilon_1}{\skind}{\upsilon_1}}$. Notice
how both bound variables $\typec\alpha$ and $\typec\beta$ are renamed to
$\typec{\upsilon_1}$, the first variable available for replacement.
Also, renaming blatantly violates the Barendregt's variable
convention~\cite{DBLP:books/daglib/0067558} used in so many works; for example
$\rename(\tapp{\upsilon_1}{(\tlambda{\alpha}{T}{\alpha})}) =
\tapp{\upsilon_1}{(\tlambda{\upsilon_1}{T}{\upsilon_1})}$, where variable
$\typec{\upsilon_1}$ is both free and bound in the resulting type.
Even if renaming violates the variable convention, substitution can still be
performed without resorting to the ``on-the-fly'' renaming of Curry and
Feys~\cite{CurryFeys58,DBLP:books/cu/HindleyS86}. When $\typec{\upsilon_1}\neq\typec{\upsilon_2}$,
%
we have
\begin{equation*}
  \tapp{(\tlambda{\upsilon_1}{\kind}{\tlambda{\upsilon_2}{\kind'}{U}})}{T} \betaarrow
(\tlambda{\upsilon_2}{\kind'}{U})\subs T{\upsilon_1} = \tlambda{\upsilon_2}{\kind'}{(U\subs T{\upsilon_1})}
\end{equation*}
since the renaming rule for application guarantees that
$\typec{\upsilon_2}\notin\fv(\TT)$. 
Otherwise if $\typec{\upsilon_1}=\typec{\upsilon_2}$, we have $(\tlambda{\upsilon_1}{\kind'}{U})\subs T{\upsilon_1} = \tlambda{\upsilon_1}{\kind'}{U}$. 
This justifies the inclusion of set $S$ in the
renaming process.
From now on, we assume that all types have gone through the renaming process.


\begin{figure}[t!]
	\declrel{Type reduction}{$\lambdared TT$}
	\begin{mathpar}
		\infer[\rseqone]{}{\lambdared{\semit \Skip T}{T}}
		
		\infer[\rseqtwo]{\lambdared{T}{V}}{\lambdared{\semit TU}{\semit VU}}
		
		\infer[\rassoc]{}{\lambdared{\semit {(\semit TU)}V}{\semit T {(\semit UV)}}}
		
		\infer[\rmu]{}{\lambdared{\tapp{\tmu{k}}\TT}{\tapp{T}({\tapp{\tmu{k}} \TT})}}

		\infer[\rbeta]{}{\lambdared{\tapp{(\tabs{\alpha}{\kind}{T})}{U}}{\blk{\rename(}\TT\subs U\alpha)}}
		
		\infer[\rtappl]{\lambdared TU}{\lambdared{TV}{UV}}

		\infer[\rdsemi]{}{\lambdared{\tapp \Dual {(\semit T U)}}{\semit {\tapp \Dual
					T}{\tapp \Dual U}}}
		
		\infer[\rdskip]{}{\lambdared {\tapp \Dual \Skip}{\Skip}}
		
		\infer[\rdend]{}{\lambdared {\tapp \Dual \End}{\End}}
		
		\infer[\rdin]{}{\lambdared{\tapp \Dual {(\tapp ?T)}}{\tapp !T}}
		
		\infer[\rdout]{}{\lambdared{\tapp \Dual {(\tapp !T)}}{\tapp ?T}}
		
		\infer[\rdextc]{}{\lambdared{\tapp {\Dual} {( {\&\{\overline{l_i:T_i}\}})}}{\oplus\{\overline{l_i:\Dual(T_i)}\}}}
		
		\infer[\rdintc]{}{\lambdared{\tapp {\Dual} {( {\oplus\{\overline{l_i:T_i}\}})}}{\&\{\overline{l_i:\Dual(T_i)}\}}}
%
		
		\infer[\rdctx]{\lambdared TU}{\lambdared{\tapp \Dual T}{\tapp \Dual U}}
		
		\infer[\rddvar]{}{\lambdared{\tapp {\Dual} {(\tapp \Dual (\alpha\ \TT_1\ldots \TT_m))}}{\alpha\ \TT_1\ldots \TT_m}}
	\end{mathpar}
	\caption{Type reduction.}
	\label{fig:reduction}
\end{figure}


Next comes the notion of \emph{type reduction} (\cref{fig:reduction}). Apart
from beta reduction (rule \rbeta), the definition provides for
%
sequential composition, for unfolding recursive types and for reducing
$\tapp \Dual {\TT}$ types. Note that renaming is further invoked in rule \rbeta
for beta reduction does not preserve renaming: consider the renamed type
$\tapp{(\tlambda{\upsilon_1}{\tkind}{\tlambda{\upsilon_2}{\tkind}{\function{\upsilon_1}{\upsilon_2}}})}{\Unit}$.
The type resulting from the substitution
$(\tlambda{\upsilon_2}{\tkind}{\function{\upsilon_1}{\upsilon_2}})\subs\Unit{\upsilon_1}$
is $\tlambda{\upsilon_2}{\tkind}{\function{\Unit}{\upsilon_2}}$ which is not
renamed and, therefore, not equivalent to
$\tlambda{\upsilon_1}{\tkind}{\function{\Unit}{\upsilon_1}}$ according to our
rules in \cref{sec:type-equiv}. 
Thanks to our modified rule \rbeta, we preserve preservation of renaming under reductions: if $\TT = \rename(\TT)$ and $\lambdared TU$ then $\UT = \rename(\UT)$.

\begin{figure}[t!]
	\declrel{Weak head normal form }{$\iswhnf T$}
	\begin{mathpar}
		\infer[\wconszero]{}{\iswhnf {\iota}}
		
		\infer[\wconsone]{\typec{\iota}\neq \typec{;}\ , \tmu{\kind}\ , \Dual \premspace m\geq 1}{\iswhnf {\tapp{\iota}{\TT_1\ldots \TT_m}}}
		
		\infer[\wseqone]{}{\iswhnf {\tapp ; T}}
		
		\infer[\wseqtwo]{\iswhnf{T_1} \\ 
			\typec{T_1} \neq \semit{U}{V}\ , \Skip \\ m\ge2}{\iswhnf{\tapp \Semi {T_1\ldots T_m}}}
		
		\infer[\wvar]{m\geq 0}{\iswhnf {\tapp{\alpha}{\TT_1\ldots \TT_m}}}
		
		\infer[\wabs]{}{\iswhnf {\tabs{\alpha}{\kind}{T}}}
		
		\infer[\wdual]{\iswhnf{T_1} \\ \typec{T_1} \neq \Skip\ , \End\ , \tapp \sharp \UT \ , \semit U V\ , \typec{\odot\{\overline{l_i: U_i}\}} \ , \tapp{\Dual}{(\alpha\ U_1 \ldots U_n)} \\ m\ge2}{\iswhnf {\tapp \Dual {\TT_1 \ldots \TT_m}}}
	\end{mathpar}
	%
		 


	\caption{Weak head normal form. 
	}
    \label{fig:whnf}
    \label{fig:normalisation}
    \label{fig:normalises}
\end{figure}

We also need the notion of \emph{weak head normal form}, $\iswhnf T$, borrowed from the lambda calculus~\cite{DBLP:books/daglib/0067558,barendregt1977type}. \cref{fig:whnf} provides a rule-based characterisation (which can be used in a compiler).
Notice that in rules \rulename{W-Seq1} and \rulename{W-Seq2}, the sequential composition operator $\semit{}{}$ appears in prefix notation. However, for type formation (kinding), it turns out that $\tapp \Semi {T_1\ldots T_m}$ is a (well formed) type only when $m\leq 2$. 
The following result shows that the rule-based characterization is equivalent to irreducibility.

\begin{lemma}
  \label{thm:whnfresults}\
    $\iswhnf T$ iff $\TT \centernot\betaarrow$.
\end{lemma}

\begin{proof}
    ($\Rightarrow$) 
    Assume that $\betared{\TT}{\UT}$ for some $\UT$. We prove that it is not the case that $\iswhnf{\TT}$.
    
    \begin{itemize}
      \item $\lambdared{\semit{\Skip}{\TT}}{\TT}$ and by inspection of $\iswhnf{}$ rules, we conclude that $\semit{\Skip}{\TT}$ is not a weak head normal form.
      
      \item $\lambdared{\semit \TT \UT}{\semit V \UT}$ only if $\lambdared{\TT}{V}$. However, by hypothesis, given that $\TT$ reduces, $\TT$ is not a weak head normal form. Therefore, none of the $\iswhnf{}$ rules can be applied. So, $\semit \TT \UT$ is not a weak head normal form.
      
      \item $\lambdared{\semit{(\semit \TT \UT)}{V}}{\semit{T}{(\semit U V)}}$ and by inspection of $\iswhnf{}$ rules, we conclude that $\semit{(\semit \TT \UT)}{V}$ is not a weak head normal form.
      
      \item $\lambdared{\tapp{\tmu{\kind}}{T}}{\tapp{T}{(\tapp{\tmu{\kind}}{T})}}$  and by inspection of $\iswhnf{}$ rules, we conclude that $\tapp{\tmu{k}}{T}$ is not a weak head normal form.
      
      \item $\lambdared{\tapp{(\tabs{\alpha}{\kind}{T})}{U}}{\renameb{T\subs{U}{\alpha}}}$ and by inspection of $\iswhnf{}$ rules, we conclude that $\tapp{(\tabs{\alpha}{\kind}{T})}{U}$ is not a weak head normal form.
      
      \item $\lambdared{T\ V}{U\ V}$ only if $\lambdared{T}{U}$.
      However, by hypothesis, given that $\TT$ reduces, $\TT$ is not a weak head normal form. Therefore, by inspection of $\iswhnf{}$ rules, we conclude that $\tapp{T}{V}$ is not a weak head normal form.
      
      \item Cases of $\TT$ for which $\lambdared{\tapp {\Dual} {\TT}}{\UT}$ for some $\UT$ are automatically excluded in rule $\wdual$, therefore none of those $\tapp \Dual \TT$ types is a weak head normal form.
    \end{itemize}
    
    ($\Leftarrow$) We must investigate all types $\TT$ such that $\TT \centernot{\betaarrow}$. We illustrate with a couple of cases:
    
    \begin{itemize}
      \item No constant $\typec{\iota}$ reduces, and according to $\wconszero$, they are all in weak head normal form. Analogously for $\typec{; \TT}$, $\tapp{\alpha}{T_1,\ldots, T_m} (m\geq 0)$ and  $\tlambda{\alpha}{\kind}{\TT}$, using rules $\wseqone$, $\wvar$ and $\wabs$, respectively.
      
      \item The cases of $\TT$ for which $\semit \TT \UT$ does not reduce are, according to rule $\rseqtwo$, those where $\TT$ does not reduce either. By induction, $\iswhnf \TT$. Also, $\TT \neq \Skip, \semit {T_1}{T_2}$, otherwise the type $\semit \TT \UT$ would reduce via $\rseqone$ or $\rassoc$, respectively. This case is covered in $\wseqtwo$ and therefore the type is in weak head normal form.
      
      \item The cases of $\TT$ for which $\typec{\Dual(T)}$ does not reduce are all such that $\TT\neq \Skip$, $\End$, $\semit{T_1}{T_2}$, $\MSGn{T}$, $\typec{\odot\{\overline{t_i:T_i}\}}$, $\typec{\Dual(\tapp{\alpha}{T_1\ldots T_m})}$ and $\TT$ does not reduce.  By induction, $\iswhnf \TT$. Each of those cases for $\TT$ is covered by $\wdual$.
      
    \end{itemize}
\end{proof}

We say that type $\TT$ \emph{normalises} to type $\UT$, written $\normalred TU$, if $\iswhnf \UT$ and $\UT$ is reached from $\TT$ in a finite number of
reduction steps (note that any term which is already whnf normalises to itself). We write $\isnormalised T$ to denote that $\normalred TU$ for some $\UT$. 

For example, suppose we want to normalise the type $\tapp{\tmu\skind}{T}$, where
$\TT$ is the type
$\tabs {\upsilon_1} \skind {\semit{\intchoice\records{\donel\colon\End,
\morel\colon\OUTn\alpha}}{\tdual {\upsilon_1}}}$.
By computing all reductions from ${\tmu{\skind}}{\TT}$, we obtain the sequence
\begin{equation*} {\tmu{\skind}}{\TT} \betaarrow \tapp \TT {(\tmu{\skind} \TT)}
  \betaarrow \semit{\intchoice\records{\donel\colon\End,
      \morel\colon\OUTn\alpha}}{\tdual (\tmu{\skind}T)} \centernot\betaarrow
\end{equation*}
Finally, we can show that $\iswhnf {\semit{\intchoice\records{\donel\colon\End, \morel\colon\OUTn\alpha}}{\tdual
  (\tmu{\skind}T)}}$:
\begin{equation*}
  \label{eq:whnf}
  \inferrule*[right= \wseqtwo]{
    \inferrule*[right= \wconsone]{
    }{
      \iswhnf{\intchoice\records{\donel\colon\End, \morel\colon\OUTn\alpha}}
    }
  }{
    \iswhnf{\semit{\intchoice\records{\donel\colon\End, \morel\colon\OUTn\alpha}}{\tdual (\tmu{\skind}T)}}
  }
\end{equation*}
Hence, we conclude that $\normalred{\tapp{\tmu\skind}{T}}{\semit{\intchoice\records{\donel\colon\End, \morel\colon\OUTn\alpha}}{\tdual
  (\tmu{\skind}T)}}$.
Similarly, we can reason that
$\tapp{\tmu{\tkind}}{(\tlambda{{\upsilon_1}}{\tkind}{{\upsilon_1}})}$,
$\tapp{\tmu{\skind}}{(\tlambda{{\upsilon_1}}{\skind}{\semit\Skip{\upsilon_1}})}$
and
$\tapp{\tmu{\skind}}{(\tlambda{{\upsilon_1}}{\skind}{\tapp\Dual{\upsilon_1}})}$
are all examples of non-normalising expressions.


\begin{figure}[t!]
  \declrel{Type formation}{$\istype \Delta T \kind$}
  \begin{mathpar}
    \infer[\kconst]{}{\istype \Delta \iota {\kind_\iota}}
    \quad
    \infer[\kvar]{\tbind\alpha\kind\in{\Delta}}
    {\istype \Delta \alpha \kind}
    \quad
    \infer[\ktabs]{\istype {\Delta+\tbind\alpha \kind} {T} {\kind'}}
    {\istype \Delta {\tabs\alpha \kind \TT} {\karrow \kind {\kind'}}}
    \quad
    \inferrule[\ktapp]{\istype \Delta {T} {\karrow \kind {\kind'}} \;\;\, \istype \Delta U \kind \;\;\, \isnormalised {\tapp TU}}
    {\istype \Delta {\tapp T U} {\kind'}}
  \end{mathpar}
  \caption{Type formation.}
  \label{fig:type-formation}
\end{figure}


Equipped with normalisation, we can introduce \emph{type formation}, which we do
via the rules in \cref{fig:type-formation}. Rule \kconst introduces constants as
types whose kinds match those of \cref{fig:constants}. Rule \kvar reads the kind
of a type variable from context $\Delta$.
An abstraction $\tabs{\alpha}{\kind}{\TT}$ is a well-formed type with kind
$\karrow{\kind}{\kind'}$ if $\TT$ is well formed in context $\Delta$ updated
with entry $\tbind{\alpha}{\kind}$ (rule \ktabs). The update is necessary since
we are dealing with renamed types and the same type variable may appear with
different kinds in nested abstractions.

It is not until we reach rule \ktapp that we find a proviso about the normalisation of a type. 
This is standard and analogous to a condition on
contractivity. 
The goal is to eliminate types that reduce indefinitely without reaching a weak head normal form.


\begin{theorem}
  \label{thm:normalisation}
  Let $\istype \Delta T \kind$.
  \begin{description}
  \item[Preservation.] If $\betared{T}{U}$, then $\istype \Delta U \kind$.
  \item[Confluence.] If $\betared{T}{U}$ and $\betared{T}{V}$, then
    $\betareds{U}{W}$ and $\betareds{V}{W}$. 
  \item[Weak normalisation.] $\normalred{T}{U}$ for some $\UT$. Furthermore,
    if $\normalred{T}{V}$, then $\typec U=\typec V$.
  \end{description}
\end{theorem}

\begin{proof}
  \textbf{Preservation:} By rule induction on $\lambdared{T}{U}$. Let us inspect some cases of (well-formed) types $\TT$ that reduce as an illustration:
  
  \begin{itemize}
    \item $\lambdared{\semit\Skip U}{U}$; from the assumption that $\istype{\Delta}{\TT}{\kind}$ follows that $\kind = \skind$ and that $\istype{\Delta}{\UT}{\skind}$.
    
    \item $\lambdared{\semit{U_1}{U_2}}{\semit V{U_2}}$ only if $\lambdared{U_1}{V}$; from the assumption that $\istype{\Delta}{\TT}{\kind}$ follows that $\kind = \skind$, $\istype{\Delta}{\UT_1}{\skind}$ and $\istype{\Delta}{\UT_2}{\skind}$. By hypothesis, $\istype{\Delta}{V}{\skind}$. Therefore, $\istype{\Delta}{\semit V{U_2}}{\skind}$.
    
    \item $\lambdared{\tmu{k}\UT}{\tapp \UT {(\tmu{k}\UT)}}$; from the assumption that $\istype{\Delta}{\TT}{\kind}$ and since $\istype{\Delta}{\tmu{\kind}}{\karrow{(\karrow{\kind}{\kind})}{\kind}}$, follows that $\istype{\Delta}{\UT}{\karrow{\kind}{\kind}}$. Therefore, $\istype{\Delta}{\tapp \UT {(\tmu{\kind}\UT)}}{\kind}$.
  \end{itemize}
  \par
  \textbf{Confluence:} By a case analysis on the various possible reductions
  from $\TT$. We sketch a simple case.
  $\lambdared{(\semit \Skip\TT);\UT}{\semit TU}$ (via $\rseqtwo$) and
  $\lambdared{(\semit \Skip\TT);\UT}{\Skip; (\semit TU)}$ (via $\rassoc$).
  Nonetheless, the latter reduces again via $\rseqone$:
  $\lambdared{\Skip; (\semit TU)}{\semit TU}$.
  \par
  \textbf{Weak normalisation:} By a case analysis on $\TT$. If $\TT=\typec{\iota},\typec{\alpha}$ or $\tabs\alpha \kind \UT$, then $\iswhnf{\TT}$ by rules $\whnfrulename{Const0}, \whnfrulename{Var}, \whnfrulename{Abs}$, respectively. Therefore $\normalred{\TT}{\TT}$ according to $\normalisationrulename{Whnf}$. If $\TT = \tapp{U}{V}$, then according to $\ktapp$ $\isnormalised{UV}$, so $\normalred{T}{T'}$ for some $\typec{T'}$ by definition.
  The second part follows from confluence.
\end{proof}

We finally arrive at the main decidability result in this section. 
In its proof, we make use of the fact that recursion is restricted to kind
$\kast$ to limit the possible subexpressions of the form
$\tapp{\tmu{\kast}}{U}$ that might appear in the normalisation of $\TT$.

\begin{restatable}[Decidability of type formation]{theorem}{kindingdecidable}
\label{thm:kindingdecidable}
  $\istype{\Delta}{\TT}{\kind}$ is decidable for types in $\FMuAstSemiOmega$. 
\end{restatable}

\begin{proof}
In \cref{app:kinding}.
\end{proof}


\section{Type equivalence} 
\label{sec:type-equiv}

This section introduces type bisimulation as our notion of type equivalence.
We define a labelled transition system (LTS) on the space of all types and write
$\ltsred TaU$ to denote that $\TT$ has a transition by label $\tsymc a$ to
$\UT$. The grammar for labels and the LTS rules are in \cref{fig:lts}.


\begin{figure}[t!]
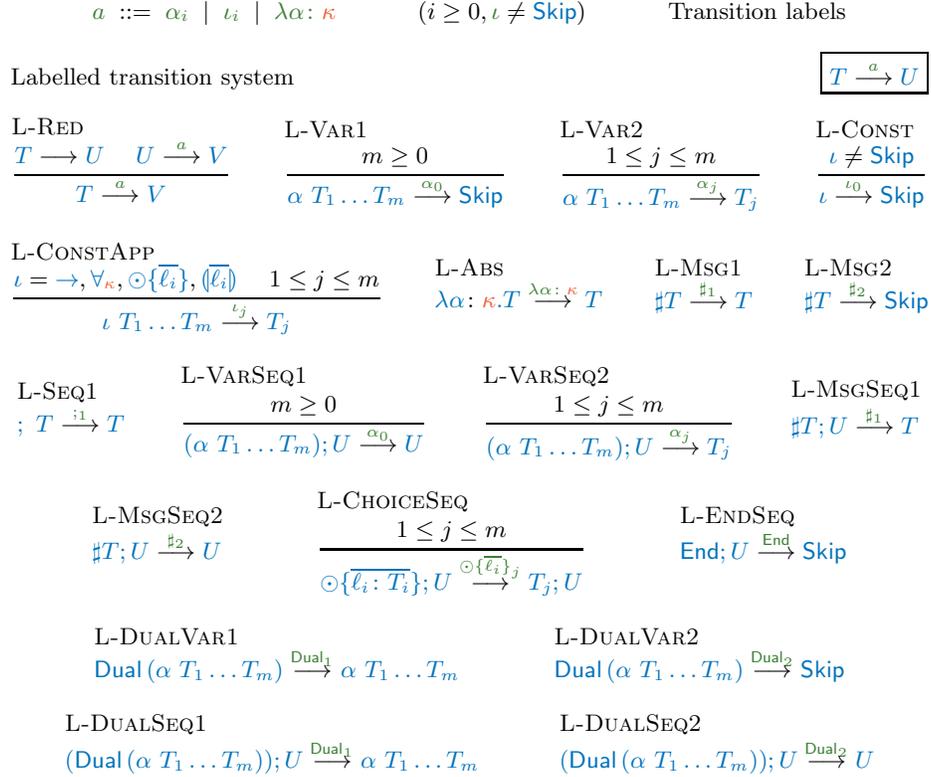

  \begin{align*}
	\tsymc a &\grmeq
     \tsymc{\alpha_i} \grmor \tsymc{\iota_i} \grmor
               \tsymc{\lambda\alpha\colon\kind}
  & (i\ge0, \tsymc\iota\neq\Skip)
  && \text{Transition labels}
  \end{align*}
  \declrel{Labelled transition system}{$\ltsred{T}{a}{U}$}
  \begin{mathpar}
  \infer[\rulename{L-Red}]{
    \lambdared TU 
    \premspace 
    \ltsred{U}{a}{V}
  }{
    \ltsred{T}{a}{V}
  }
  
  \infer[\rulename{L-Var1}]{
    m\geq 0
  }{
    \ltsred{\alpha\ T_1\ldots T_m}{\alpha_0}\Skip
  }
  
  \infer[\rulename{L-Var2}]{
    1\leq j\leq m
  }{
    \ltsred{\alpha\ T_1\ldots T_m}{\alpha_j}{T_j}
  }
  
  \infer[\rulename{L-Const}]{
    \tiota \neq \Skip
  }{
    \ltsred\iota{\iota_0}\Skip
  }
  
  \infer[\rulename{L-ConstApp}]{
    \typec \iota= \function{}{}, \typec{\forall}_\kind, \typec{\odot{\{\overline{\ell_i}\}}}, \typec{\varrecs{\overline{\ell_i}}}
    \premspace
    1\leq j\leq m
  }{
    \ltsred{\iota\ T_1\ldots T_m}{\iota_j}{T_j}
  }
  
  \infer[\rulename{L-Abs}]{
    \ltsred{\tlambda \alpha\kind T}{\lambda\alpha\colon\kind}{T}
  }{}
  
  \infer[\rulename{L-Msg1}]{
    \ltsred{\MSGn T}{\sharp_1}{T}
  }{}
  
  \infer[\rulename{L-Msg2}]{
    \ltsred{\MSGn T}{\sharp_2}{\Skip}
  }{}
  
  \infer[\rulename{L-Seq1}]{
    \ltsred{;\ T}{;_1}{T}
  }{}
  
  \infer[\rulename{L-VarSeq1}]{
    m\geq 0
  }{
    \ltsred{\semit{(\alpha\ T_1\ldots T_m)}U}{\alpha_0}{U}
  }
  
  \infer[\rulename{L-VarSeq2}]{
    1\leq j\leq m
  }{
    \ltsred{\semit{(\alpha\ T_1\ldots T_m)}U}{\alpha_j}{T_j}
  }
  
  \infer[\rulename{L-MsgSeq1}]{
    \ltsred{\semit{\MSGn T}U}{\sharp_1}{T}
  }{}
  
  \infer[\rulename{L-MsgSeq2}]{
    \ltsred{\semit{\MSGn T}U}{\sharp_2}{U}
  }{}
  
  \infer[\rulename{L-ChoiceSeq}]{
    1\leq j\leq m
  }{
    \ltsred{\semit{\choice\records{\overline{\ell_i\colon T_i}}}U}{\tsymc{\odot{\{\overline{\ell_i}\}}_j}}{\semit{T_j}U}
  }
  
  \infer[\rulename{L-EndSeq}]{
    \ltsred{\semit\End U}{\keyword{End}}{\Skip}
  }{}
  
  \infer[\rulename{L-DualVar1}]{
    \ltsred{\tdual{(\alpha\ T_1\ldots T_m)}}{\Duall_1}{\alpha\ T_1\ldots T_m}
  }{}
  
  \infer[\rulename{L-DualVar2}]{
    \ltsred{\tdual{(\alpha\ T_1\ldots T_m)}}{\Duall_2}{\Skip}
  }{}
  
  \infer[\rulename{L-DualSeq1}]{
    \ltsred{\semit{(\tdual{(\alpha\ T_1\ldots T_m)})} U}{\Duall_1}{\alpha\ T_1\ldots T_m}
  }{}
  
  \infer[\rulename{L-DualSeq2}]{
    \ltsred{\semit{(\tdual{(\alpha\ T_1\ldots T_m)})} U}{\Duall_2}{U}
  }{}
  \end{mathpar}
  \caption{Labelled transition system for types.}
  \label{fig:lts}
\end{figure}




If $\TT$ is not in weak head normal form, then we must normalise it to some type $\UT$, so that $\TT$ has the same transitions as $\UT$ (rule \rulename{L-Red}). 
Otherwise if $\iswhnf T$, then the transitions of $\TT$ can be immediately derived by looking at the corresponding rule for $\TT$ as follows. 
If $\TT$ is a variable, use rule \rulename{L-Var1} (with $m=0$). 
If $\TT$ is a constant (other than $\Skip$), use rule \rulename{L-Const}. 
Note that if $\TT$ is a lone $\Skip$, it has no transitions. 
If $\TT$ is an abstraction, use rule \rulename{L-Abs}. 

If $\TT$ is an application, then we need to look inside the head. We write $\TT$
as $\typec{T_0\ T_1\ldots T_m}$ with $m\geq 1$ where $\typec{T_0}$ is not an
application, and look at $\typec{T_0}$. If $\typec{T_0}$ is a variable, use
rules \rulename{L-Var1} and \rulename{L-Var2}. If $\typec{T_0}$ is one of the
constants $\function{}{}$, $\typec{\forall}_\kind$, $\typec{\odot\{\overline{l_i}\}}$ or
$\typec{\varrecs{\overline{l_i}}}$, use rule \rulename{L-ConstApp}. Note that if $\typec{T_0}$
is an abstraction or $\tmu\kind$, then $\TT$ is not in weak head normal form and
it must be normalised. If $\typec{T_0}$ is $\typec\sharp$, we use rules
\rulename{L-Msg1} and \rulename{L-Msg2}. If $\typec{T_0}$ is $\Dual$, then the
only way for $\TT$ to be well-formed and in weak head normal form is if $m=1$
and $\typec{T_1}$ is $\talpha$ or $\typec{\alpha\ U_1\ldots U_m}$, in which case
we use rules \rulename{L-DualVar1} and \rulename{L-DualVar2}.

If $\typec{T_0}$ is $\typec;$ , we require an additional case analysis on
$\typec{T_1}$. If $m=1$, use rule \rulename{L-Seq1}. Otherwise $m=2$ due to
kinding. If $\typec{T_1}$ is a variable, use rule \rulename{L-VarSeq1} (with
$m=0$). If $\typec{T_1}$ is a constant, then it must be of kind $\skind$, so it
must be either $\Skip$ or $\End$. If $\typec{T_1}$ is $\Skip$, then $\TT$ is not
in weak normal form and it must be normalised. If $\typec{T_1}$ is $\End$, use
rule \rulename{L-EndSeq} ($\End$ is an absorbing element, so $\semit\End\UT$ simply makes a transition to $\Skip$ without executing $\UT$). Note that $\typec{T_1}$ cannot be an abstraction due
to kinding.

If $\typec{T_1}$ is an application, then again we write $\typec{T_1}$ as $\typec{U_0\ U_1\ldots U_n}$ with $n\geq 1$ where the head $\typec{U_0}$ is not an application, and look at $\typec{U_0}$.
If $\typec{U_0}$ is a variable, use rules \rulename{L-VarSeq1} and \rulename{L-VarSeq2}.
If $\typec{U_0}$ is a constant, it must be one of $\typec;$ , $\tmu\kind$, $\typec\sharp$, $\typec{\odot\{\overline{l_i}\}}$ or $\typec\Dual$ due to kinding.
If $\typec{U_0}$ is $\typec\sharp$, use rules \rulename{L-MsgSeq1} and \rulename{L-MsgSeq2}. 
If $\typec{U_0}$ is $\typec{\odot\{\overline{l_i}\}}$, use rule \rulename{L-ChoiceSeq}. 
If $\typec{U_0}$ is $\Dual$, the only way for $\TT$ to be well-formed and in weak head normal form is if $n=1$ and $\typec{U_1}$ is $\talpha$ or $\typec{\alpha\ V_1\ldots V_\ell}$, in which case we use rules \rulename{L-DualSeq1} and \rulename{L-DualSeq2}. 
Finally if $\typec{U_0}$ is $\typec;$ , $\tmu\kind$ or an abstraction, $\TT$ is not in weak normal form and it must be normalised.\smallskip


\begin{figure}[t!]
\begin{tikzpicture}[thick, node distance = 0.8cm, every node/.style = {font = \small}]
\node at (0.6,1.1) (T) {$\tabs {\upsilon_1} \tkind U$};
\node at (2.6,1.1) (U) {$\UT$};
\node at (4.6,0) (Un)
{$\semit{\intchoice\records{\donel\colon\End, \morel\colon\OUTn{\upsilon_1}}}{\tdual \UT}$};
\node at (5.6,1.1) (EDU)
{$\semit\End{\tdual U}$};
\node at (7.8,1.1) (S1) {$\Skip$};
\node at (5.5,-1.1) (SaDU)
{$\semit{\OUTn{\upsilon_1}}{\tdual U}$};
\node at (6.6,-2.2) (a) {$\typec{\upsilon_1}$};
\node at (6.6,-3.4) (S2) {$\Skip$};
\node at (4.6,-2.2) (DU) {$\tdual U$};
\node at (0,-2.2) (DUn)
{$\semit{\extchoice\records{\donel\colon\End, \morel\colon\INn{\upsilon_1}}}{\tdual{(\tdual \UT)}}$};
\node at (-0.6,-3.4) (EDDU)
{$\semit\End{\tdual{(\tdual U)}}$};
\node at (-3.3,-3.4) (S3) {$\Skip$};
\node at (-0.6,-1.1) (RaDDU)
{$\semit{\INn{\upsilon_1}}{\tdual{(\tdual U)}}$};
\node at (-1.6,0) (a2) {$\typec{\upsilon_1}$};
\node at (-1.6,1.1) (S4) {$\Skip$};
\node at (0,0) (DDU) {$\tdual{(\tdual U)}$};

\draw[->] (T) -- node [above] {\scriptsize $\tsymc{\lambda{\upsilon_1}\colon\tkind}$} (U);
\draw[double, ->] (U) -- (Un);
\draw[->] (Un) -- node [right] {\scriptsize $\tsymc{\intchoicep\{\dones,\mores\}_1}$}(EDU);
\draw[->] (EDU) -- node [above] {\scriptsize $\tsymc{\Endl}$} (S1);
\draw[->] (Un) -- node [right] {\scriptsize $\tsymc{\intchoicep\{\dones,\mores\}_2}$}(SaDU);
\draw[->] (SaDU) -- node [right] {\scriptsize $\tsymc{!_1}$} (a);
\draw[->] (a) -- node [right] {\scriptsize $\tsymc{{\upsilon_1}}$} (S2);
\draw[->] (SaDU) -- node [right] {\scriptsize $\tsymc{!_2}$} (DU);
\draw[double, ->] (DU) -- (DUn);
\draw[->] (DUn) -- node [left] {\scriptsize $\tsymc{\extchoicep\{\dones,\mores\}_1}$} (EDDU);
\draw[->] (EDDU) -- node [above] {\scriptsize $\tsymc{\Endl}$} (S3);
\draw[->] (DUn) -- node [left] {\scriptsize $\tsymc{\extchoicep\{\dones,\mores\}_2}$} (RaDDU);
\draw[->] (RaDDU) -- node [left] {\scriptsize $\tsymc{?_1}$} (a2);
\draw[->] (a2) -- node [left] {\scriptsize $\tsymc{{\upsilon_1}}$} (S4);
\draw[->] (RaDDU) -- node [left] {\scriptsize $\tsymc{?_2}$} (DDU);
\draw[double, ->] (DDU) -- (Un);

\end{tikzpicture}
\caption{The LTS for type $\tabs {\upsilon_1} \tkind U$. Normalisation
  $\normalred{T_1}{T_2}$ is represented as $\typec{T_1}\Rightarrow\typec{T_2}$ and $\UT$ is a shorthand for type
  $\tmu\skind\, \typec(\tabs {\upsilon_2} \skind
  {\semit{\intchoice\records{\donel\colon\End, \morel\colon\OUTn{\upsilon_1}}}{\tdual
      {\upsilon_2}}}\typec)$.}
\label{fig:lts-example}
\end{figure}


Let us clarify our LTS rules with an example. Consider the following type
$\tabs {\upsilon_1} \tkind {\tmuinfix {\upsilon_2} \skind
  {\semit{\intchoice\records{\donel\colon\End,
        \morel\colon\OUTn{\upsilon_1}}}{\tdual {\upsilon_2}}}}$ and call it
$\TT$. Clearly $\TT$ is of kind $\karrow\tkind\skind$.
For a given functional type $\typec{\upsilon_1}$, this type specifies a channel that alternates between: offer a choice and output a value of type $\typec{\upsilon_1}$; or select a choice and input a value of type $\typec{\upsilon_1}$. The polarity is swapped thanks to the application of constant $\Dual$ to the recursion variable $\typec{\upsilon_2}$.
To construct the (fragment of the) LTS generated by this type, let us first
desugar $\TT$ into $\tabs {\upsilon_1} \tkind U$ where $\UT$ is the type $\tmu\skind\,\typec(\tabs {\upsilon_2} \skind {\semit{\intchoice\records{\donel\colon\End, \morel\colon\OUTn{\upsilon_1}}}{\tdual {\upsilon_2}}}\typec)$.
%
%
Notice that $\UT$ normalises to
${\semit{\intchoice\records{\donel\colon\End, \morel\colon\OUTn{\upsilon_1}}}{\tdual \UT}}$.
The LTS for the example is sketched in \cref{fig:lts-example}. In this case,
only finitely many types appear. However, more elaborate examples involving
sequential composition or higher-order recursion may lead to an infinite graph
of transitions.\smallskip

Given the LTS rules, we can define, in the standard way, a notion of bisimulation.
\label{def:bisimulation}
  A binary relation $R$ on types is called a \emph{bisimulation} if, for every
  $(\TT,\UT)\in R$ and every transition label $\tsymc a$:
\begin{enumerate}
	\item if $\ltsred Ta{T'}$, then there exists $\typec{U'}$ \st $\ltsred Ua{U'}$ and $(\typec{T'},\typec{U'})\in R$;
	\item if $\ltsred Ua{U'}$, then there exists $\typec{T'}$ \st $\ltsred Ta{T'}$ and $(\typec{T'},\typec{U'})\in R$.
\end{enumerate}
We say that types $\TT$ and $\UT$ are bisimilar, written $\isbisim TU$, if there exists a bisimulation $R$ such that $(\TT,\UT)\in R$. 

Intuitively, a notion of type equivalence must preserve and reflect the syntax of type constructors: for example, a type $\function TU$ is equivalent to a type $\function{T'}{U'}$ iff $\TT$, $\typec{T'}$ are equivalent and $\UT$, $\typec{U'}$ are equivalent. 
Using the bisimulation technique, we achieve this by considering a labelled transition system on types: $\function TU$ has a transition labelled $\tsymc{\rightarrow_1}$ to $\TT$ and a transition labelled $\tsymc{\rightarrow_2}$ to $\UT$. 
In this way, $\function TU$ can only be equivalent to another type which has two transitions with those same labels. 
For each of the type constructors (
$\typec\rightarrow$, $\typec{\forall_\kind}$, 
$\typec{!}$, $\typec{?}$, $\typec{\odot{\{\overline{\ell_i}\}}}$,
and so on) we have suitable transition rules. 
Moreover, a type sometimes needs to be reduced before a type constructor is found at the root of the syntax tree. 
If $\TT$ normalizes to $\UT$, then we expect $\TT$ and $\UT$ to be bisimilar, which is achieved thanks to rule \rulename{L-Red}. 
This handles the various reductions: 
beta-reductions arising from lambda-abstraction and applications (e.g., $\tapp{(\tabs\alpha\kind T)}{U}$ reduces to $\subs\UT\alpha \TT$), 
reductions arising from the monoidal structure of sequential composition (e.g., $\semit\Skip\TT$ reduces to $\TT$), 
reductions arising from the internalization of duality as a type constructor (e.g., $\tdual{(\OUTn\TT)}$ reduces to $\INn\TT$) 
and reductions arising from the recursion (e.g., $\tapp{\mu_\kind}\TT$ reduces to $\tapp \TT {(\tapp{\mu_\kind}\TT)}$).

Let us look at some examples of bisimilarity.

\begin{itemize}
\item Consider the types $\TT = \typec{\alpha\ (! \Int)}$ and
  $\UT = \semit{\alpha\ (! \Int)}{\Skip}$. These are bisimilar, since they both
  exhibit a transition with label $\tsymc{\alpha_1}$ to $\OUTn\Int$ (rules
  \rulename{L-Var2} and \rulename{L-VarSeq2} resp.) and another transition with
  label $\tsymc{\alpha_0}$ to $\Skip$ (rules \rulename{L-Var1} and
  \rulename{L-VarSeq1} resp.). In general, two types $\talpha\ \TT$ and
  $\semit{\alpha\ \UT}{\VT}$ are bisimilar iff $\isbisim{\TT}{\UT}$ and
  $\isbisim{\Skip}{\VT}$. This justifies our choice for the rules with
  variables.
\item \sloppy{Similar considerations apply to $\MSGn T$ vs $\semit {\MSGn T} \Skip$ and
  to $\tdual{(\alpha\ T_1\ldots T_m)}$ vs
  $\semit{\tdual{(\alpha\ T_1\ldots T_m)}}{\Skip}$.}
\item For internal and external choices, the situation is slightly different.
  Types $\typec{\odot\records{\gol\colon\semit TV, \quitl\colon\semit UV}}$ and
  $\semit{\odot\records{\gol\colon T, \quitl\colon U}}V$ are bisimilar. The
  former exhibits transitions $\tsymc{\odot\{\gos,\quits\}_1}$ and $\tsymc{\odot\{\gos,\quits\}_2}$ to
  $\semit TV$ and $\semit UV$ resp., due to rule \rulename{L-ConstApp}. The
  latter exhibits the same transitions, but due to rule \rulename{L-ChoiceSeq}
  instead. This is as desired (distributivity of sequential composition against
  choice) and justifies our choice for those rules.
\item $\End$ is bisimilar to any type $\semit \End U$, due to our choice of rules \rulename{L-Const} and \rulename{L-EndSeq}.
\item $\Skip$ is the only well-formed type in weak head normal form that has no transitions.
  Hence, if $\istype \Delta T \kind$, we have $\isbisim T \Skip$ iff $\normalred{T}{\Skip}$.
\item Consider the types $\tabs\alpha\kind T$ and $\tabs\beta\kind U$. 
Before generating the LTS, we must rename them to $\tabs{\upsilon_i}\kind{\subs{\upsilon_i}\alpha T}$ and $\tabs{\upsilon_j}\kind{\subs{\upsilon_j}\beta U}$. 
According to rule \rulename{L-Abs}, the renamed types are bisimilar only if $\typec{\upsilon_i}=\typec{\upsilon_j}$, in which case we check the equivalence of $\subs{\upsilon_i}\alpha \TT$ and $\subs{\upsilon_j}\beta \UT$.
Here, the choice of $\typec{\upsilon_i}$ depends only on $\TT$ (not on $\UT$). 
This allows us to define an LTS for $\TT$ (independently of $\UT$), and an LTS for $\UT$ (independently of $\TT$), in a way that the types are equivalent iff the corresponding LTS are bisimilar. 
Without using renaming, we would have to somehow ``remember'' that variables $\typec\alpha$, $\typec\beta$ should be ``linked'', which would make the generation of independent LTS impossible. 
This is the main technical issue solved by our notion of renaming.
\end{itemize}

Our notion of type equivalence enjoys natural properties and behaves as expected with respect to the notions of reduction, normalisation and kinding from \cref{sec:types}. 
We can derive rules for type equivalence, that could be used to define another coinductive notion of equivalence, via effective syntax-directed rules. 
In other words, one could prove that $\sim$ is the greatest relation satisfying properties 1-15 below.

\begin{lemma}
    \ 
    \begin{enumerate}
    \item $\isbisim T\alpha$ iff $\normalred T\alpha$.
    \item $\isbisim T{\tabs\alpha\kind U}$ iff $\normalred T{\tabs\alpha\kind V}$ for some $\VT$ \st $\isbisim UV$.
    \item $\isbisim T\End$ iff
        \begin{itemize}
        \item $\normalred T\End$; or
        \item $\normalred T{\semit\End U}$ for some $\UT$.
        \end{itemize}
    \item Let $\tiota$ be a constant other than $\End$. Then $\isbisim T\iota$ iff $\normalred T\iota$.
    \item $\isbisim T{\alpha\ U_1\ldots U_m}$ iff:
        \begin{itemize}
        \item $\normalred T{\alpha\ U'_1\ldots U'_m}$ for some $\typec{U'_1},\ldots,\typec{U'_m}$ \st $\isbisim{U_j}{U'_j}$ for each $j$; or
        \item $\normalred T{\semit{(\alpha\ U'_1\ldots U'_m)}W}$ for some $\typec{U'_1},\ldots,\typec{U'_m},\VT$ \st $\isbisim{U_j}{U'_j}$ and $\isbisim \Skip V$.
        \end{itemize}
    \item $\isbisim T{\function UV}$ iff $\normalred T{\function{U'}{V'}}$ for some $\typec{U'},\typec{V'}$ \st $\isbisim U{U'}$ and $\isbisim V{V'}$.
    \item $\isbisim T{\forall_\kind U}$ iff $\normalred T{\forall_\kind U'}$ for some  $\typec{U'}$ \st $\isbisim U{U'}$.
    \item $\isbisim T{\varrecf \ell {U_\ell} L}$ iff $\normalred T{\varrecf \ell {U'_\ell} L}$ for some $\typec{U'_\ell}$ \st $\isbisim{U_\ell}{U'_\ell}$ for each $\ell$.
    \item $\isbisim T{\MSGn U}$ iff:
        \begin{itemize}
        \item $\normalred T{\MSGn{U'}}$ for some $\typec{U'}$ \st $\isbisim U{U'}$; or
        \item $\normalred T{\semit{\MSGn{U'}}V}$ for some $\typec{U'},\VT$ \st $\isbisim U{U'}$ and $\isbisim\Skip V$.
        \end{itemize}
    \item $\isbisim T{\odot\recordf \ell {U_\ell} L}$ iff:
        \begin{itemize}
        \item $\normalred T{\odot\recordf \ell {U'_\ell} L}$ for some $\typec{U'_\ell}$ \st $\isbisim {U_\ell}{U'_\ell}$ for each $\ell$; or
        \item $\normalred T{\semit{\odot\recordf \ell {U'_\ell} L}V}$ for some $\typec{U'_\ell},\VT$ \st $\isbisim {U_\ell}{\semit{U'_\ell}V}$ for each $\ell$.
        \end{itemize}
    \item $\isbisim T{\semit\End U}$ iff
        \begin{itemize}
        \item $\normalred T\End$; or
        \item $\normalred T{\semit\End U'}$ for some $\typec{U'}$.
        \end{itemize}
    \item $\isbisim T{\semit{\MSGn U}V}$ iff:
        \begin{itemize}
        \item $\normalred T{\MSGn{U'}}$ for some $\typec{U'}$ \st $\isbisim U{U'}$ and $\isbisim V\Skip$; or
        \item $\normalred T{\semit{\MSGn{U'}}{V'}}$ for some $\typec{U'},\typec{V'}$ \st $\isbisim U{U'}$ and $\isbisim V{V'}$.
        \end{itemize}
    \item $\isbisim T{\semit{\odot\recordf \ell {U_\ell} L}V}$ iff:
        \begin{itemize}
        \item $\normalred T{\odot\recordf \ell {U'_\ell} L}$ for some $\typec{U'_\ell}$ \st $\isbisim {\semit{U_\ell}V}{U'_\ell}$ for each $\ell$; or
        \item $\normalred T{\semit{\odot\recordf \ell {U'_\ell} L}V'}$ for some $\typec{U'},\typec{V'}$ \st $\isbisim {\semit{U_\ell}V}{\semit{U'_\ell}V}$ for each $\ell$.
        \end{itemize}
    \item $\isbisim T{\tdual{(\alpha\ U_1\ldots U_m)}}$ iff
        \begin{itemize}
        \item $\normalred T{\tdual{(\alpha\ U'_1\ldots U'_m)}}$ for some $\typec{U'_j}$ \st $\isbisim {U_j}{U'_j}$ for each $j$; or
        \item $\normalred T{\semit{\tdual{(\alpha\ U'_1\ldots U'_m)}}V}$ for some $\typec{U'_j},\VT$ \st $\isbisim {U_j}{U'_j}$ for each $j$ and $\isbisim\Skip V$.
        \end{itemize}
    \item $\isbisim T{\semit{\tdual{(\alpha\ U_1\ldots U_m)}}V}$ iff
        \begin{itemize}
        \item $\normalred T{\tdual{(\alpha\ U'_1\ldots U'_m)}}$ for some $\typec{U'_j}$ \st $\isbisim {U_j}{U'_j}$ for each $j$ and $\isbisim V\Skip$; or
        \item $\normalred T{\semit{\tdual{(\alpha\ U'_1\ldots U'_m)}}V'}$ for some $\typec{U'_j},\typec{V'}$ \st $\isbisim {U_j}{U'_j}$ for each $j$ and $\isbisim V{V'}$.
        \end{itemize}
    \end{enumerate}
    \end{lemma}

Next, we show that type equivalence is preserved under renaming, reduction and normalization. 

\begin{lemma}\label{lem:bisimprop}
\ 
\begin{enumerate}
\item\label{lem:bisimrename} If $\isbisim TU$, then $\isbisim{\renameb{T}}{\renameb{U}}$.
\item\label{lem:bisimred} If $\betared TU$, then $\isbisim TU$. If $\normalred TU$, then $\isbisim TU$.
\item\label{lem:bisimnormal} Suppose that $\normalred T{T'}$ and $\normalred U{U'}$. Then $\isbisim TU$ iff $\isbisim {T'}{U'}$.
\end{enumerate}
\end{lemma}

\begin{proof}
\cref{lem:bisimrename} follows by coinduction on $\isbisim{T}{U}$, \ie we prove that the set $\{(\renameb{\TT},\renameb{\UT}) : \isbisim TU\}$ is a bisimulation. 
The only interesting case is that in which $\TT=\tabs\alpha\kind{T'}$ and $\UT=\tabs\beta\kindp{U'}$. 
Since $\isbisim TU$, we get that $\talpha=\tbeta$, $\kind=\kindp$ and $\isbisim{T'}{U'}$. 
It is straightforward to see that $\isbisim{T}{U}$ implies $\fv(\TT)=\fv(\UT)$ since free variables are eventually `seen' in the process of consuming $\TT$, $\UT$. 
Therefore, $\first(\TT)=\first(\UT)=\typec{\upsilon_k}$ for some $k$, and $\renameb{\TT}=\tabs{\upsilon_k}\kind{\renameb{T'\subs{\upsilon_k}{\alpha}}}$, $\renameb{\UT}=\tabs{\upsilon_k}\kind{\renameb{U'\subs{\upsilon_k}{\alpha}}}$. 
Finally, we apply \rulename{L-Abs} to the pair $(\renameb{\TT},\renameb{\UT})$, arriving at the pair $(\renameb{\typec{T'\subs{\upsilon_k}{\alpha}}},\renameb{\typec{U'\subs{\upsilon_k}{\alpha}}})$.
This lies in our set since $\isbisim{T'}{U'}$ implies $\isbisim{T'\subs{\upsilon_k}{\alpha}}{U'\subs{\upsilon_k}{\alpha}}$.

The first part of \cref{lem:bisimred} is straightforward: if $\betared TU$, then by \rulename{L-Red} every transition $\ltsred TaV$ is matched by a transition $\ltsred UaV$ and vice-versa. To prove the second part of \cref{lem:bisimred}, assume that $\normalred TU$ and consider a finite sequence of reductions $\TT=\betared{T_0}{T_1}\betared{}{\ldots}\betared{}{T_n}=\UT$. At each step of the sequence we have $\isbisim{T_k}{T_{k+1}}$ and thus (since $\isbisim{}{}$ is an equivalence relation) we get $\TT=\isbisim{T_0}{T_{n}}=\UT$.

\cref{lem:bisimnormal} is a direct consequence of \cref{lem:bisimred}. $\normalred T{T'}$ and $\normalred U{U'}$imply $\isbisim T{T'}$ and $\isbisim U{U'}$. Since $\isbisim{}{}$ is an equivalence relation, we get $\isbisim TU$ iff $\isbisim {T'}{U'}$.

\end{proof}

Finally, we show that the axioms for sequential composition in the introduction \eqref{eq:axiomssemi} are derivable from our notion of bisimulation. 

\begin{lemma}\label{lem:axsemi}
\ 
\begin{enumerate}
  \item\label{lem:axsemineutral} $\isbisim{\semit\Skip T}{T}$ (Neutral element).
  \item\label{lem:axsemiabsorb} $\isbisim{\semit\End T}{\End}$ (Absorbing element).
  \item\label{lem:axsemiassoc} $\isbisim{\semit{(\semit TU)}{V}}{\semit{T}{(\semit UV)}}$ (Associativity).
  \item\label{lem:axsemidistr} $\isbisim{\semit{\tchoice{l_i}{T_i}} U}{\tchoice{l_i}{\semit{T_i}U}}$ (Distributivity).
\end{enumerate}
\end{lemma}

\begin{proof}
\cref{lem:axsemineutral} follows from the observation that $\lambdared{\semit \Skip T}{T}$ (\rseqone), together with \cref{lem:bisimred} of \cref{lem:bisimprop}.
\cref{lem:axsemiabsorb}, follows from the observation that $\ltsred{\semit\End T}{\Endl}{\Skip}$ (\rulename{L-EndSeq}) and $\ltsred{\End}{\Endl}{\Skip}$ (\rulename{L-Const}). 
\cref{lem:axsemiassoc} follows from the observation that $\lambdared{\semit{(\semit TU)}{V}}{\semit{T}{(\semit UV)}}$ (\rassoc), together with \cref{lem:bisimred} of \cref{lem:bisimprop}.
\cref{lem:axsemidistr} follows from the observation that $\ltsred{\semit{\tchoice{l_i}{T_i}} U}{\tsymc{\odot\{\overline{l_i}\}_j}}{\semit{T_j}U}$ (\rulename{L-ChoiceSeq}) and $\ltsred{\tchoice{l_i}{\semit{T_i}U}}{\tsymc{\odot\{\overline{l_i}\}_j}}{\semit{T_j}U}$ (\rulename{L-ConstApp}).
\end{proof}



\section{Decidability of type equivalence}
\label{sec:equivalence-decidable}

This section presents results on decidability of type equivalence. Our approach
consists in translating types to objects in some computational model. We look at
finite-state automata (for types in $\FMu$, $\FMuAstOmega$, $\FMuDot$, and
$\FMuAstDotOmega$), simple grammars (for types in $\FMuSemi$ and
$\FMuAstSemiOmega$) and deterministic pushdown automata (for types in
$\FMuOmega$, $\FMuDotOmega$ and $\FMuSemiOmega$).

  \label{def:grammar}
  We say that a \emph{grammar in Greibach normal form} is a tuple
  $(\T, \N, \nsymc\gamma, 
  \R)$ where: $\T$ is a set of terminal symbols, denoted by
  $\tsymc{a},\tsymc{b},\tsymc{c}$; $\N$ is a set of nonterminal symbols, denoted
  by $\nsymc{X}, \nsymc{Y}, \nsymc{Z}$; $\nsymc\gamma\in \N^\ast$ is the
  starting word;
and $\R \subseteq \N \times \T \times \N^\ast$ is a set of productions. A
grammar is said to be \emph{simple} if, for every nonterminal $\nsymc{ X}$ and every
terminal $\tsymc{a}$, there is at most one production
$(\nsymc{X},\tsymc{a},\nsymc\delta)\in\R$ \cite{DBLP:conf/focs/KorenjakH66}.

Greek letters $\nsymc\gamma$ and $\nsymc\delta$ denote (possibly empty) words of
nonterminal symbols.
Productions are written as $\gltsred Xa\delta$.
We define a notion of bisimulation for grammars via a
labelled transition system. The system comprises a set of states $\N^\ast$
corresponding to words of nonterminal symbols. For each production
$\gltsred Xa\gamma$ and each word of nonterminal symbols
$\nsymc{\delta}$, we have a labelled transition
$\gltsred{X\delta}{a}{\gamma\delta}$. We let $\gequiv$ denote the bisimulation
relation for grammars (the definition is similar to that in \cref{def:bisimulation}).

For the moment we focus on the class $\FMuAstSemiOmega$ and we explain how to
convert a type $\TT$ into a simple grammar
$(\T_{\typec{T}},\N_{\typec{T}},\word(\TT),\R_{\typec{T}})$. 
The conversion is based
on a function $\word(\TT)$ that maps each type $\TT$ into a word of nonterminal
symbols, while introducing fresh nonterminals and productions. In our construction, following the approach by Costa \etal~\cite{DBLP:journals/corr/abs-2203-12877}, we use a nonterminal symbol with no productions, denoted by $\nsymc\bot$, in order to separate the two descendants of a send/receive operation such as $\semit{\OUTn\TT}{\UT}$.
The sequence of nonterminal symbols $\word(\TT)$ is defined as follows. 
First consider the cases in which $\iswhnf\TT$.
\begin{itemize}
\item For any $m\geq 0$:
  $\word(\talpha\ \typec{T_1}\typec{\ldots}\typec{T_m})=\Ynt$ for $\Ynt$ a fresh
  nonterminal symbol with a production $\gltsred Y{\alpha_0}{\emptyword}$ as well as
  $\gltsred Y{\alpha_j}{\wordb{\typec{T_j}}\bot}$ for each $1\leq j\leq m$.
\item $\word(\Skip)=\emptyword$.
\item $\word(\End)=\Ynt$ for $\Ynt$ a fresh symbol with a single
  production $\gltsred Y\Endl{\bot}$.
\item for any $\tiota\neq\Skip,\End$: $\word(\tiota) = \Ynt$ for $\Ynt$ a fresh
  nonterminal symbol with a single production $\gltsred Y\iota\emptyword$.
\item $\word(\tlambda{\alpha}\kind T) = \Ynt$ for $\Ynt$ a fresh 
  symbol with a production $\gltsred Y {\lambda\alpha\colon\kind}{\wordb{\TT}}$.
\item for any $m\geq 1$ and for $\tiota$ one of $\function{}{}$,
  $\tforall\kind$, $\typec{\choice\{\overline{l_i}\}}$, $\typec{\varrecs{\overline{l_i}}}$:
  $\word(\typec{\iota\ T_1\cdots T_m})=\Ynt$ for a fresh nonterminal
  $\Ynt$ with a production $\gltsred Y{\iota_j}{\wordb{\typec{T_j}}}$ for each
  $1\leq j\leq m$.
\item $\word(\MSGn\TT)=\Ynt$ for $\Ynt$ a fresh symbol with productions $\gltsred Y{\tsymc{\sharp_1}}{\wordb{\TT}\bot}$ and $\gltsred Y{\sharp_2}{\emptyword}$.
\item $\word(\Semi\ \TT)=\Ynt$ for $\Ynt$ a fresh symbol with a production $\gltsred Y{;_1}{\wordb{\TT}}$.
\item $\word(\semit \TT \UT) = \word(\TT)\word(\UT)$.
\item $\word(\tdual{(\talpha\ \typec{T_1}\typec{\ldots}\typec{T_m})}) = \Ynt$ for $\Ynt$ a fresh symbol with productions $\gltsred Y{\Duall_1}{\wordb{\talpha\ \typec{T_1}\typec{\ldots}\typec{T_m}}}$ and $\gltsred Y{\Duall_2}\emptyword$.
\end{itemize}
Finally, let us handle the cases where $\TT$ is not in weak head normal form.
\begin{itemize}
\item If $\normalred T\Skip$, then $\word(\TT)=\emptyword$.
\item Otherwise if $\normalred TU\neq\Skip$, then $\word(\TT)=\Ynt$ for $\Ynt$ a fresh nonterminal symbol. Let $\Znt\ntdelta = \word(\UT)$. Then $\Ynt$ has a production $\gltsred Y a {\gamma\delta}$ for each production $\gltsred Z a \gamma$.
\end{itemize}

In the above construction, we create fresh symbols each time we encounter a weak head normal form other than $\Skip$. In other words, $\N_\TT$ is the set containing $\nsymc\bot$ and all nonterminals $\Ynt$ created during the computation of $\word(\TT)$. Another key insight is that the sequential composition of types is translated into a concatenation of words: $\word(\typec{T_1\Semi T_2\Semi\ldots\Semi T_n}) = \word(\typec{T_1})\word(\typec{T_2})\ldots\word(\typec{T_n})$. This allows our construction to terminate: even if the transitions lead to infinitely many types, they are split on the sequential composition operator, and so we only need to consider finitely many subexpressions.

For the last case in our construction to be well-defined, \ie when $\normalred TU\neq\Skip$, we require $\word(\UT)$ to be non-empty. Indeed, if $\iswhnf\UT$, then we can observe (by inspecting all cases) that $\word(\UT)=\emptyword$ iff $\UT=\Skip$.

We also need to argue that the construction of $\word(\TT)$ eventually terminates. For this, we keep track of all types visited during the construction, and we only add a fresh nonterminal $\Ynt$ to our grammar if the type visited is syntactically different from all types visited so far. Therefore, we reuse the same symbol $\Ynt$ with the same productions each time we revisit a type. With all these observations, we get the following result.

\begin{lemma}
\label{lem:sgtermination}
Suppose that $\TT\in \FMuAstSemiOmega$. Then the construction of $\word(\TT)$ terminates producing a simple grammar.
\end{lemma}

We illustrate the above construction with the polymorphic tree exchanging
example from \cref{sec:motivation},
\begin{lstlisting}
type TreeC a = &{Leaf: Skip, Node: TreeC a; ?a ; TreeC a}
\end{lstlisting}
that is written in $\FMuAstSemiOmega$ as
$
\typec{T_0} = \tabs{\upsilon_1}\tkind{\tmuinfix{\upsilon_2}\skind{\extchoice\records{\leafl\colon\Skip,\nodel\colon\semit{{\upsilon_2}}{\semit{\INn{\upsilon_1}}{{\upsilon_2}}}}}}.
$
Since $\typec{T_0}$ is in weak head normal form, $\word(\typec{T_0})$ returns a
fresh symbol, which we call $\nsymc{X_0}$. We also have a production
$\gltsred{X_0}{\lambda{\upsilon_1}\colon\kind}{\wordb{\typec{T_1}}}$, where
$\typec{T_1}$ is the type $\tmuinfix{\upsilon_2}\skind{\extchoice\records{\leafl\colon\Skip,\nodel\colon\semit{{\upsilon_2}}{\semit{\INn{\upsilon_1}}{{\upsilon_2}}}}}$.
Since $\typec{T_1}$ is not in whnf, we must normalise it, to get $\typec{T_2} = \typec{\extchoice\records{\leafl\colon\Skip,\nodel\colon\semit{T_1}{\semit{\INn {\upsilon_1}}{T_1}}}}$. Therefore $\word(\typec{T_1})$ returns a fresh symbol, which we  call $\nsymc{X_1}$. To obtain the transitions of $\nsymc{X_1}$, we must first compute $\word(\typec{T_2})$, which is a fresh symbol $\nsymc{X_2}$ with transitions $\gltsred{X_2}{\&_1}{\wordb{\Skip}}$ and $\gltsred{X_2}{\&_2}{\wordb{\semit{T_1}{\semit{\INn{\upsilon_1}}{T_1}}}}$. Thus we also get $\gltsred{X_1}{\&_1}{\wordb{\Skip}}$ and $\gltsred{X_1}{\&_2}{\wordb{\semit{T_1}{\semit{\INn{\upsilon_1}}{T_1}}}}$.

\begin{sloppypar}
  We have $\word(\Skip)=\emptyword$, but we still need to compute
  $\word(\semit{T_1}{\semit{\INn{\upsilon_1}}{T_1}})$. This type normalises to
  $\typec{T_3}=\semit{T_2}{\semit{\INn{\upsilon_1}}{T_1}}$ since
  $\normalred{T_1}{T_2}$. Thus
  $\word(\semit{T_1}{\semit{\INn{\upsilon_1}}{T_1}})$ is a fresh symbol
  $\nsymc{X_3}$. To obtain the productions of $\nsymc{X_3}$ we must compute
  $\word(\semit{T_2}{\semit{\INn{\upsilon_1}}{T_1}}) =
  \word(\typec{T_2})\word(\INn{\upsilon_1})\word(\typec{T_1})$. At this point we
  already have $\word(\typec{T_1})=\nsymc{X_1}$ and
  $\word(\typec{T_2})=\nsymc{X_2}$.
\end{sloppypar}
We still need to compute $\word(\INn{\upsilon_1})$, which is a fresh symbol $\nsymc{X_4}$ with productions $\gltsred{X_4}{?_1}{\wordb{\typec{\upsilon_1}}\bot}$ and $\gltsred{X_4}{?_2}{\emptyword}$. In turn, $\word(\typec{\upsilon_1})$ is a fresh symbol $\nsymc{X_5}$ with a production $\gltsred{X_5}{{\upsilon_1}}{\emptyword}$. Finally, we get $\word(\semit{T_2}{\semit{\INn{\upsilon_1}}{T_1}}) = \nsymc{X_2X_4X_1}$, which means we can write the productions for $\nsymc{X_3}$: $\gltsred{X_3}{\&_1}{X_4X_1}$ and $\gltsred{X_3}{\&_2}{X_3X_4X_1}$.

Putting all this together, we can finally obtain the simple grammar:
\begin{align*}
\gltsred{X_0}{\lambda{v_1}\colon\tkind}{X_1}
&&
\gltsred{X_1}{\&_1}{\emptyword}
&&
\gltsred{X_1}{\&_2}{X_3}
&&
\gltsred{X_2}{\&_1}{\emptyword}
&&
\gltsred{X_2}{\&_2}{X_3}
\\
\gltsred{X_3}{\&_1}{X_4X_1}
&&
\gltsred{X_3}{\&_2}{X_3X_4X_1}
&&
\gltsred{X_4}{?_1}{X_5\bot}
&&
\gltsred{X_4}{?_2}{\emptyword}
&&
\gltsred{X_5}{v_1}{\emptyword}
\end{align*}

Next, we argue that type equivalence (\ie bisimilarity on types) corresponds to bisimilarity on the corresponding grammars. This is achieved by the following lemma, that asserts that the LTS of a type and the LTS of the corresponding word of nonterminals have exactly the same transitions.

\begin{lemma}[Full abstraction]
\label{lem:fullyabstracttranslation}
Let $\TT\in \FMuAstSemiOmega$ and $(\T_{\typec{T}},\N_{\typec{T}},\word(\TT),\R_{\typec{T}})$ the corresponding simple grammar. Suppose also that $\isbisimg{\wordb{\TT}}\gamma$.
\begin{enumerate}
  \item If $\ltsred TaU$ then there exists $\nsymc{\gamma'}$ such that $\gltsred \gamma a {\gamma'}$ and $\isbisimg{\wordb{\UT}}{\gamma'}$.
  \item If $\gltsred \gamma a {\gamma'}$ then there exists $\UT$ such that $\ltsred TaU$ and $\isbisimg{\wordb{\UT}}{\gamma'}$.
\end{enumerate}
\end{lemma}

\begin{proof}
The proof is by coinduction, showing that the relation
$$\R = \{(\TT,\nsymc\gamma)\colon\isbisimg{\wordb\TT}\gamma\}$$
is backward closed for the transition relation. This is done by a case analysis on $\TT$. First, assume that $\iswhnf\TT$. We show a few cases.
\begin{itemize}
\item If $\TT$ is $\tapp\talpha{\typec{T_1}\typec{\ldots}\typec{T_m}}$, then by rules \rulename{L-Var1} and \rulename{L-Var2} the LTS at $\TT$ has transitions $\ltsred{\TT}{\alpha_0}{\Skip}$ and $\ltsred{\TT}{\alpha_j}{T_j}$ for $1\leq j\leq m$. 
Similarly, $\wordb{\TT}=\Ynt$ for $\Ynt$ with productions $\gltsred Y{\alpha_0}{\emptyword}$ as well as $\gltsred Y{\alpha_j}{\wordb{\typec{T_j}}\bot}$ for $1\leq j\leq m$. 
Since $\isbisimg\Ynt\gamma$, the LTS at $\nsymc\gamma$ has transitions $\gltsred\gamma{\alpha_0}{\gamma_0}$ and $\gltsred\gamma{\alpha_j}{\gamma_j}$ for some $\nsymc{\gamma_0}$ \st $\isbisimg{\emptyword}{\gamma_0}$ and $\nsymc{\gamma_j}$ \st $\isbisimg{\wordb{\typec{T_j}}\bot}{\gamma_j}$ for $1\leq j\leq m$. 
Since $\wordb{\Skip}=\emptyword$, we get that $(\Skip,\nsymc{\gamma_0})\in\R$. 
Finally, since $\isbisimg{\wordb{\typec{T_j}}}{\wordb{\typec{T_j}}\bot}$, we get that $(\typec{T_j},\nsymc{\gamma_j})\in\R$ for $1\leq j\leq m$.
\item If $\TT$ is $\Skip$, then the LTS at $\TT$ has no transitions. 
Similarly, $\wordb{\TT}=\emptyword$ has no transitions. 
Since $\isbisimg\emptyword\gamma$, the LTS at $\nsymc\gamma$ has no transitions either.
\item If $\TT$ is $\semit UV$, then by rule \rulename{W-Seq2} $\iswhnf U$ and $\UT$ is neither $\Skip$ nor a sequential composition. 
We must perform a second case analysis on $\UT$. 
For example if $\UT$ is $\tapp\talpha{\typec{U_1}\typec{\ldots}\typec{U_m}}$, then by rules \rulename{L-VarSeq1} and \rulename{L-VarSeq2} the LTS at $\TT$ has transitions $\ltsred{\TT}{\alpha_0}{V}$ and $\ltsred{\TT}{\alpha_j}{T_j}$ for $1\leq j\leq m$. 
Similarly, $\wordb{\TT}=\wordb\UT\wordb{\VT}=\Ynt\wordb{\VT}$ for $\Ynt$ with productions $\gltsred Y{\alpha_0}{\emptyword}$ as well as $\gltsred Y{\alpha_j}{\wordb{\typec{U_j}}\bot}$ for $1\leq j\leq m$. 
Since $\isbisimg{\Ynt\wordb{\VT}}\gamma$, the LTS at $\nsymc\gamma$ has transitions $\gltsred\gamma{\alpha_0}{\gamma_0}$ and $\gltsred\gamma{\alpha_j}{\gamma_j}$ for some $\nsymc{\gamma_0}$ \st $\isbisimg{\wordb{\VT}}{\gamma_0}$ and $\nsymc{\gamma_j}$ \st $\isbisimg{\wordb{\typec{U_j}}\bot\wordb{\VT}}{\gamma_j}$ for $1\leq j\leq m$. 
We immediately get that $(\VT,\nsymc{\gamma_0})\in\R$. 
Finally, since $\isbisimg{\wordb{\typec{U_j}}}{\wordb{\typec{U_j}}\bot\wordb{\VT}}$, we get that $(\typec{U_j},\nsymc{\gamma_j})\in\R$ for $1\leq j\leq m$.
\end{itemize}
Next, assume that $\TT$ is not in whnf. There are two possibilities.
\begin{itemize}
\item If $\normalred T\Skip$, then by rule \rulename{L-Red} the LTS at $\TT$ has no transitions.
Similarly, $\word(\TT)=\emptyword$ has no transitions. 
Since $\isbisimg\emptyword\gamma$, the LTS at $\nsymc\gamma$ has no transitions either.
\item If $\normalred TU\neq\Skip$, then by rule \rulename{L-Red} the LTS at $\TT$ has a transition $\ltsred T a {U'}$ iff the LTS at $\UT$ has a corresponding transition $\ltsred U a {U'}$ (with the same $\typec{U'}$).
Similarly, $\word(\TT)=\Ynt$ for $\Ynt$ with productions $\gltsred Y a {\gamma'\delta}$ where $\Znt\ntdelta = \word(\UT)$ and $\Znt$ has productions $\gltsred Z a {\gamma'}$. 
Therefore, the LTS at $\wordb\TT$ has a transition $\gltsred{\wordb\TT}{a}{\delta'}$ iff the LTS at $\wordb{\UT}$ has a transition $\gltsred{\wordb{\UT}}{a}{\delta'}$ (with the same $\nsymc{\delta'}$). 
Hence $\isbisimg{\wordb\UT}{\wordb\TT}\isbisimg{}{\gamma}$, so that $(\UT,\nsymc{\gamma})\in\R$.
Finally, since $\iswhnf U$ the previous case analysis shows that $\UT$ and $\nsymc\gamma$ have matching transitions.
We conclude that $\TT$ and $\nsymc\gamma$ have matching transitions as well.
\end{itemize}
\end{proof}

As a consequence of the above result, we get soundness and completeness of the bisimilarity $\isbisimg{\wordb{\TT}}{\wordb{\UT}}$ with respect to the bisimilarity $\isbisim TU$. Indeed by \cref{lem:fullyabstracttranslation}, any sequence of transitions starting from $\TT$ can be matched by a sequence of transitions starting from $\word(\TT)$; and similarly for $\UT$. Thus $\isbisim TU$ iff $\isbisimg{\wordb{\TT}}{\wordb{\UT}}$.

\begin{restatable}{theorem}{decidabilitytypeequiv}
	\label{thm:decidability-type-equiv}
  The type equivalence problem is decidable for types in $\FMuAstSemiOmega$.
\end{restatable}

\begin{proof}
    Given types $\TT$, $\UT$ in $\FMuAstSemiOmega$, an algorithm for deciding
    $\isbisim TU$ is as follows. First compute $\word(\TT)$ and $\word(\UT)$; this
    terminates due to \cref{lem:sgtermination}. Then decide whether
    $\isbisimg{\wordb{\TT}}{\wordb{\UT}}$, using any known algorithm for
    bisimilarity of simple grammars, \eg Almeida \etal\cite{DBLP:conf/tacas/AlmeidaMV20} or Burkart \etal\cite{DBLP:conf/mfcs/BurkartCS95}, whose time complexity is doubly-exponential. Correctness follows from the discussion immediately preceding this theorem.
\end{proof}

For the remainder of this section, we look at the other classes of types in \cref{fig:fmudiagram} and examine the computation models they correspond to. Since class $\FMuSemi$ is contained in $\FMuAstSemiOmega$, we can express types without $\lambda$-abstractions with simple grammars as well. In this way we recover previous results in the literature \cite{DBLP:conf/tacas/AlmeidaMV20,DBLP:journals/corr/abs-2203-12877}.

Let us now look at the class $\FMuAstDotOmega$. In this class we do not have $\Skip$ nor sequential composition and message operators are binary ($\MSG T U$) rather than unary. 
Since we do not have sequential composition, there is no need to consider words of nonterminals, and instead it suffices to translate types into single symbols, \ie states in an automata. 
Moreover, since there is no recursion beyond $\tmu\kind$, only finitely many types can be reached from a given $\TT$. We could thus adapt our construction as follows for $\FMuAstDotOmega$.

In the definition of the LTS (\cref{fig:lts}):
\begin{itemize}
  \item discard all rules involving sequential composition;
  \item discard rules \rulename{L-Var1} for $m>0$ and \rulename{L-DualVar2} (they were only needed to distinguish types in sequential composition);
  \item discard case $\tiota=\End$ in rule \rulename{L-Const} (so that $\End$ no longer has transitions);
  \item replace $\Skip$ with $\End$ on the right-hand side of rules \rulename{L-Var1} with $m=0$ and \rulename{L-Const};
  \item discard rules \rulename{L-Msg1} and \rulename{L-Msg2} and treat $\tiota=\typec\sharp$ like the other constants in rule \rulename{L-ConstApp}.
\end{itemize}

We also replace the construction of $\word(\TT)$ into a construction of $\tstate(\TT)$, associating to each type $\TT$ a state in a finite-state automata. For each transition $\ltsred TaU$ we have the corresponding transition $\ltsred {\tstateb\TT}a{\tstateb\UT}$. Notice that the resulting automata is deterministic since the original LTS is also deterministic (for each type $\TT$ and label $\tsymc a$, there is at most one transition $\ltsred TaU$). Since bisimilarity of deterministic finite-state automata can be decided in polynomial time \cite{hopcroftkarp:1971}, we get the following results.

\begin{theorem}~
\label{thm:equivalenceFmucdotkast}
\begin{enumerate}
\item To each type $\TT$ in $\FMuAstDotOmega$ we can associate a finite-state automata corresponding to the (fragment of the) LTS generated by $\TT$.
\item The type equivalence problem is polynomial-time decidable for types in $\FMuAstDotOmega$. 
\end{enumerate}
\end{theorem}

Clearly, \cref{thm:equivalenceFmucdotkast} applies to the subclasses of
$\FMuAstDotOmega$: $\FMu$, $\FMuDot$ and $\FMuAstOmega$.
In this way we recover previous results in the literature
\cite{DBLP:conf/popl/CaiGO16,DBLP:journals/corr/abs-2203-12877,DBLP:conf/fossacs/GayPV22}.

Finally, we consider the classes $\FMuOmega$, $\FMuDotOmega$ and $\FMuOmega$ involving arbitrarily-kinded recursion. We shall show that these classes are already powerful enough to simulate deterministic pushdown automata; hence, the type equivalence problem becomes impractical (\ie no practical implementation of an algorithm is known). We only focus on the simplest case $\FMuOmega$, as the others two classes are even more expressive. Instead of looking at deterministic pushdown automata, we instead look at deterministic first-order grammars, which are an equivalent model of computation~\cite{DBLP:journals/corr/abs-1010-4760}. The advantage of considering deterministic first-order grammars is to simplify our construction.
%
We say that a \emph{first-order grammar} is a tuple $(\X, \T, \N, \nsymc{E}, \R)$ where:
\begin{itemize}
\item $\X$ is a set of variables $\nsymc\alpha, \nsymc\beta, \ldots$; $\T$ is a set of terminal symbols $\tsymc a, \tsymc b, \ldots$; $\N$ is a set of nonterminal symbols $\Xnt, \Ynt,\ldots$
\item each nonterminal $\Xnt$ has an arity $m=\arity(\Xnt)\in \nbb$.
\item the set $\E$ of expressions over $\X$, $\N$ is inductively defined by two rules: any variable $\nsymc\alpha$ is an expression; if $\arity(\Xnt)=m$ and $\nsymc{E_1},\ldots,\nsymc{E_m}$ are expressions, then so is $\nsymc{X\ E_1 \ldots E_m}$. Whenever $m=0$, $\Xnt$ is called a constant.
\item $\Ent$ is an expression over $\N$, called the initial expression.
\item $\R$ is a set of productions. Each production is a triple
  $(\Xnt, \tsymc a, \Ent)$, written as
  $\gltsred{X\ \alpha_1 \ldots \alpha_m}aE$, where $m=\arity(\Xnt)$ and the
  variables in $\Ent$ must be taken from
  $\nsymc{\alpha_1},\ldots,\nsymc{\alpha_m}$.
\end{itemize}
A first-order grammar is \emph{deterministic} if, for every $\Xnt$ and $\tsymc a$, there is at most one production $(\Xnt, \tsymc a, \Ent)\in\R$.

Just as a simple grammar defines an LTS over words of nonterminals, a first-order grammar defines an LTS over the set $\E_0$ of closed expressions. For each production $\gltsred{X\ \alpha_1 \ldots \alpha_m}aE$ we have the labelled transition $\gltsred{X\ E_1 \ldots E_m}a{E\blk[\nsubs{E_1}{\alpha_1}\blk,\ldots\blk,\nsubs{E_m}{\alpha_m}\blk]}$.

Let $\isbisimg{\!\!}{\!\!}$ denote bisimilarity over closed expressions
according to a first-order grammar. We now present a fully abstract (\ie
preserving bisimilarity) translation of a deterministic first-order grammar into
a type in $\FMuOmega$. Each grammar variable $\nsymc\alpha$ has a corresponding
type variable $\talpha$ (of kind~$\tkind$). An expression
$\nsymc{X\ E_1 \ldots E_m}$ is represented as a type application
$\typec{X\ E_1 \ldots E_m}$. If $\Xnt$ has arity $m$ and the productions
$\gltsred{X\ \alpha_1 \ldots \alpha_m}{a_j}{E_j}$ for a range of $j$, then we
write the equation specifying $\XT$ as a record (since the first-order grammar
is deterministic, all record labels are distinct, and thus the right-hand side
on the equation specifying $\XT$ is well-formed).
\vspace*{-0.9ex} \begin{equation*} \iseqt X
  {\tabs{\alpha_1}\tkind{\ldots\tabs{\alpha_m}\tkind{\records{a_1\colon
          E_1,\ldots,a_m\colon E_m}}}}
\end{equation*}
This gives rise to a system of equations $\{\iseqt {X_i} {T_i}\}$, one for each nonterminal $\nsymc{X_i}$, where the nonterminals may appear in the right-hand sides $\typec{T_i}$. 
Finally, given an initial expression $\Ent$, it is standard how to convert it into a $\mu$-type using the  system above.

Using the above translation, we are able to simulate a transition $\gltsred E{a_j}F$
of the first-order grammar as a transition $\ltsred E{\records{\overline{a_i}}_j}F$ on the corresponding types. Therefore, the translation is fully abstract and we get the following result.

\begin{theorem}
  Let $\Ent$ and $\Fnt$ be closed expressions on a first-order grammar and
  $\typec E,\typec F$ the corresponding types. Then
  $\isbisimg\Ent\Fnt$ iff $\isbisim EF$.
\end{theorem}

Let us work on an example to better understand the above translation. Consider the language $L_3 = \{\ell^nar^na \mid n\geq 0\} \cup \{\ell^nbr^nb \mid n\geq 0\}$ over the alphabet $\{a,b,\ell,r\}$. $L_3$ is a typical example of a language that cannot be described with a simple grammar, but can be accepted by a deterministic pushdown automaton~\cite{DBLP:conf/focs/KorenjakH66}. Consider the  first-order grammar with nonterminals $\Xnt, \Rnt, \Ant, \Bnt, \nsymc\bot$, initial expression $\Xnt\ \Ant\ \Bnt$, and productions
\begin{align*}
\gltsred {X\ \alpha\ \beta}\ell{X\ (R\ \alpha)\ (R\ \beta)}
&&
\gltsred {X\ \alpha\ \beta}a{\alpha}
&&
\gltsred {X\ \alpha\ \beta}b{\beta}
\\
\gltsred {R\ \alpha}r\alpha
&&
\gltsred Aa\bot
&&
\gltsred Bb\bot
\end{align*}
Note that $\nsymc\bot$ is a constant without productions. It is easy to see that the traces of this first-order grammar correspond exactly to the words in $L_3$. By following the steps in the above translation, we arrive at the  system of equations
\begin{align*}
&\iseqt X { \tabs \alpha \tkind {\tabs \beta \tkind {\records{\ell\colon X(R\alpha)(R\beta), a\colon\alpha, b\colon\beta}}}}
&&
\iseqt R { \tabs \alpha \tkind {\records{r\colon\alpha}}}
\\
&\iseqt A {\records{a\colon\bot}}
&&
\iseqt B {\records{b\colon\bot}}
&&
\iseqt \bot {\records{}}
\end{align*}
Therefore, the initial expression $\nsymc{X\ A\ B}$ becomes the type
$$\typec{(\tmuinfix{\xi}{\karrow\tkind{\karrow\tkind\tkind}}{\tabs\alpha\tkind{\tabs\beta\tkind{\records{\ell\colon \xi\records{r\colon\alpha}\records{r\colon\beta}, a\colon\alpha, b\colon\beta}}})\records{a\colon\records{}}\records{b\colon\records{}}}},$$
whose transitions simulate the transitions of the first-order grammar.


\section{The term language and its metatheory}
\label{sec:expressions}
\label{sec:processes}


\begin{figure}[t!]
\begin{align*}
  \termc v \grmeq\ & \termc c
  \grmor \termc x
  \grmor \eabs xTt
  \grmor \erec xTv
  \grmor \etabs \alpha \kind v
  \grmor \erecord {l_i}{v_i}
  \grmor \evariant lvT
  \\ & \etapp \receivek T
  \grmor \etapp {\etapp \receivek T} T
  \grmor \etapp \sendk T
  \grmor \eapp {\etapp \sendk T} v
  \grmor \etapp {\eapp {\etapp \sendk T} v} T
  \\ 
  \termc t \grmeq\ & \termc v
  \grmor \eapp tt
  \grmor \etapp tT
  \grmor \erecord {l_i}{t_i}
  \grmor \elet {l_i}{x_i} tt
  \\ &  \evariant ltT
  \grmor \ecase tt
  \grmor \ematch tt
  \\
  \termc p \grmeq\ & \thread t
  \grmor \termc p \PAR \termc p
  \grmor \NU xxp
  \\
  \termc E \grmeq\ & \termc{[]}
  \grmor \eapp Et
  \grmor \eapp v E
  \grmor \etapp ET
  \grmor \termc{\{l_1=v_1,\dots, l_j=E, \dots,l_n=e_n\}}
  \\ & \elet{l_1}{x_i}Et
  \grmor \evariant lET
  \grmor \ecase Et
  \grmor \ematch Et
\end{align*}
  \begin{align*}
    \termc c \grmeq& &&&& \text{Term constant}\\
	&\receivek && \foralltinfix \alpha \tkind {\foralltinfix \beta \skind
                  {\function{\IN\alpha\beta}{\pairt\alpha\beta}}} &&
                                                                     \text{receive
                                                                     on a channel}\\
	&\sendk && \foralltinfix \alpha \tkind {\function
               \alpha{\foralltinfix \beta \skind {\function{\OUT\alpha\beta}\beta}}} && \text{send on a channel}\\
	&\selectc{l_j}{\tinttchoice{l_i}{T_i}} && \function{\tinttchoice{l_i}{T_i}}{T_j} && \text{internal choice}\\
    &\closek && \function \End \Unit && \text{channel close}\\
	&\forkk && \function {(\function\Unit\Unit)} \Unit && \text{fork a new thread}\\
	&\newk && \foralltinfix \alpha \skind {\pairt \alpha {\tdual\alpha}}  && \text{channel creation}
  \end{align*}
  \caption{Terms, and typed for term constants.}
  \label{fig:term-constants}
\end{figure}


This section briefly introduces a concurrent functional language equipped with
$\FMuAstSemiOmega$ types, together with its metatheory. The results mostly follow from those in the literature, although explicit recursion at the term level and the unrestricted bindings in typing contexts are somewhat new in session types.

%
The syntax of values, terms, processes and call-by-value evaluation contexts are
defined by the grammar in \cref{fig:term-constants}.  The same figure introduces
types for the constants.
The term language is essentially the polymorphic lambda calculus with support
for session operators, formulated as in Almeida \etal and Cai \etal~\cite{DBLP:journals/corr/abs-2106-06658,DBLP:conf/popl/CaiGO16}.
From System $F$ it comprises terms and type abstractions, records and variants,
including constructors and destructors in each case.
The support for session operations and concurrency includes channel creation
($\newk$), the different channel operations ($\receivek$, $\sendk$, $\matchk$,
$\selectk$ and $\closek$) and thread creation ($\forkk$).
We program at the term level and use processes only for the runtime. Processes
include terms as threads, parallel composition and channel creation, all
inspired in the pi-calculus with double
binders~\cite{DBLP:journals/iandc/Vasconcelos12}. 



\begin{figure}[t!]
  \declrel{Term typing}{$\isterm \Delta \Gamma t T$}
  \begin{mathpar}
    \infer[\tconstr]{\istype \Delta {T_c} \kast
    }{
      \isterm \Delta \emptyCtx c {T_c}
    }

    \infer[\tvarr]{}{\isterm \Delta {\ebind xT} x T}

    %
    \infer[\ttappr]{
      \isterm \Delta {\Gamma_1} {t_1}{U \rightarrow T}
      \\
      \isterm \Delta {\Gamma_2} {t_2}{U}
    }{
      \isterm \Delta {\Gamma_1,\Gamma_2} {\eapp{t_1}{t_2}}{T}
    }

    \infer[\trecr]{
      \istype \Delta T \kast
      \\
      \isterm \Delta {\Gamma,\eubind x {T\rightarrow U}} {v} {T \rightarrow U}
    }{
      \isterm \Delta {\Gamma} {\erec x {T\rightarrow U} v} {T\rightarrow U}
    }

    \infer[\ttabsr]{
      \isterm {\Delta,\tbind \alpha \kind} \Gamma v {\tapp T \alpha}
    }{
      \isterm \Delta \Gamma {(\etabs \alpha \kind v)} {\tapp{\tforall\kind}{T}}
    }

    \infer[\trecordr]{
      \overline{\isterm \Delta {\Gamma_i} {t_i} {T_i}}
    }{
      \isterm \Delta {\overline\Gamma} {\erecord{l_i}{t_i}} {\trecord{l_i}{T_i}}
    }

    \infer[\tprojr]{
      \isterm \Delta {\Gamma_1} {t_1} {\trecord{l_i}{T_i}}
      \\
      \isterm \Delta {\Gamma_2,\overline{\ebind{x_i}{T_i}}} {t_2} {T}
    }{
      \isterm \Delta {\Gamma_1, \Gamma_2} {\elet{l_i}{x_i}{t_1}{t_2}} {T}
    }

    \infer[\tmatchr]{
      \isterm \Delta {\Gamma_1} {t_1} {\typec{\&\{\overline{l_i:\TT_i}\}}}
      \quad
      \isterm \Delta {\Gamma_2} {t_2}{\trecord{l_i}{T_i\rightarrow T}}
    }{
      \isterm \Delta {\Gamma_1, \Gamma_2} {\ematch{t_1}{t_2}} {T}
    }
    \quad
    %
    \infer[\teqr]{
      \isterm \Delta \Gamma t U
      \quad
      \istype \Delta U  \kast
      \quad
      \isbisim U T
    }{
      \isterm \Delta \Gamma t T
    }
    
    \infer[\tderr]{
      \isterm \Delta {\Gamma,\ebind xT} t U
    }{
      \isterm \Delta {\Gamma,\eubind xT} t U
    }

    \infer[\tweakr]{
      \isterm \Delta \Gamma t U
    }{
      \isterm \Delta {\Gamma,\eubind xT} t U
    }

    \infer[\tcontrr]{
      \isterm \Delta {\Gamma,\eubind yT,\eubind zT} {t} U
    }{
      \isterm \Delta {\Gamma,\eubind xT} {t\vsubs xy\vsubs xz} U
    }
  \end{mathpar}

  \declrel{Process typing}{$\isproc \Gamma p$}
  \begin{mathpar}
    \infer{\isterm \Empty \Gamma t \Unit}{\isproc \Gamma {\thread t}}

    \infer{\isproc{\Gamma_1}{p_1} \\
      \isproc{\Gamma_2}{p_2}}{\isproc{\Gamma_1,\Gamma_2}{p_1 \PAR p_2}}

    \infer{\isproc{\Gamma,\ebind xT,\ebind y{\tapp \Dual
          T}}{p}}{\isproc{\Gamma}{\NU xyp}}
  \end{mathpar}
  (\tappr, \tabsr: see Cai \etal\cite{DBLP:conf/popl/CaiGO16}; \tcaser adapt from \tmatchr)
  \caption{Typing.}
  \label{fig:typing}
\end{figure}


Term and process typing are in \cref{fig:typing}.
A judgement of the form $\isterm \Delta \Gamma t T$ records the fact that term
$\termc t$ has type $\typec T$ under contexts $\Delta$ (recording the kinds of
type variables) and $\Gamma$ (recording types for term variables).
The judgement for processes, $\isproc \Gamma p$, says that $\termc p$ is
well-typed under context $\Gamma$. The judgement simplifies that for terms, for
processes feature no free type variables and are assigned no particular type.
Here $\Unit$ is short for the empty record type
$\typec{\{\}}$, and $\pairt TU$ is short for the record type
$\typec{\records{\fstl\colon T,\sndl\colon U}}$.
Once again, the rules are adapted from the two above cited works. The difference
to Cai \etal~\cite{DBLP:conf/popl/CaiGO16} is that we work with in a linear
setting and hence axioms (\tconstr and \tvarr) work on an empty context, and
most of the other rules must split the context accordingly. Rule \ttabsr
simplifies that of Cai \etal~\cite{DBLP:conf/popl/CaiGO16}; we can easily show
that the two rules are interchangeable.
We support exponentials \cite{DBLP:journals/tcs/Girard87} for recursive
functions, so that one may write functions that feature more than one recursive
call (good for consuming binary trees, for example) and branches that do not use
the recursive function (for code that is supposed to terminate). Towards this
end, we add an unrestricted binding $\eubind xT$ in term variable contexts, an
explicit rule for $\reck$ (as opposed to making $\reck$ a constant as in Cai \etal~\cite{DBLP:conf/popl/CaiGO16}) and structural rules for unrestricted bindings
(\tderr, \tweakr and \tcontrr).
Thanks to the power of System $F$, most of the session and concurrency operators
are expressed as constants. For example, $\receivek$ receives a session type
$\OUT\alpha\beta$ with $\typec\alpha$, the payload of the message, an arbitrary type
and $\typec\beta$, the continuation, a session type, and returns a pair of the value
received and the continuation channel. As usual  $\foralltinfix \alpha \kind
T$ abbreviates the type $\tapp{\tforall \kind}{(\tabs \alpha \kind T)}$. The
exception is the external choice (\tmatchr) which can not be captured by a type
(similarly to \tcaser) and hence requires a dedicated typing rule.



\begin{figure}[t!]
  \declrel{Term reduction}{$\expred tt$}
  \begin{mathpar}
    \infer{}{\expred{(\eapp{\eabs xTt)}v}{\termc t \vsubs vx}}
    \quad
    \infer{}{\expred{(\etapp{\etabs \alpha \kind t)}T}{\termc t \subs T \alpha}}
    \quad
    \infer{}{\expred{\elet{l_i}{x_i}{\erecord{l_i}{v_i}}{t}}{\termc t\overline{\vsubs{v_i}{x_i}}}}

    \infer{}{\expred{\ecase{(\evariant{l_j}vT)}{\erecord{l_i}{t_i}}}{\eapp{t_j}v}}


    \infer{}{\expred{\eapp{(\erec xTv)}{u}}{\eapp{(v\vsubs{\erec xTv}{x})}{u}}}

    \infer{\expred{t_1}{t_2}}{\expred{E[t_1]}{E[t_2]}}
  \end{mathpar}
  \declrel{Structural congruence (congruence rules omitted)}{$\isscong pp$}
  \begin{gather*}
    \isscong{p_1 \PAR p_2}{p_2 \PAR p_1}
    \qquad
    \isscong{(p_1 \PAR p_2) \PAR p_3}{p_1 \PAR (p_2 \PAR p_3)}
    \\
    \isscong{\NU xy{p_1} \PAR p_2}{\NU xy{(p_1 \PAR p_2)}}
    \qquad
    \isscong{\NU wx{\NU yzp}}{\NU yz{\NU wxp}}
  \end{gather*}
  \declrel{Process reduction}{$\procred pp$}
  \begin{mathpar}
    \infer{\expred{t_1}{t_2}}{\procred{\thread{t_1}}{\thread{t_2}}}

    \infer{}{\procred{\thread{E[\eapp \forkk v]}}{\thread{E[\eunit]} \PAR
        \thread{\eapp v \eunit}}}

    \infer{}{\procred{\thread{E[\etapp \newk T]}}{\NU xy {\thread{E[\epair xy]}}}}
    
    \infer{}{\procred{\NU xy{(\thread{E_1[\eapp {\eapp {\etapp {\etapp
                  \receivek T} U} y}]} \PAR \thread{E_2[\eapp {\eapp {\etapp {\etapp
                  \sendk V} W} v} x]})}}{\NU
        xy {(\thread{E_1[\epair yv]} \PAR \thread{E_2[x]})}}}

    \infer{}{\procred{\NU xy{(\thread{E_1[\ematch y{\erecord{l_i}{t_i}}]} \PAR
          \thread{E_2[\eapp {(\selectc{l_j}{T})}x]})}}{\NU xy {\thread{E_1[\eapp
            {t_j} y]} \PAR \thread{E_2[x]}}}}

    \infer{}{\procred{\NU xy{(\thread{E_1[\eapp\closek y]}\PAR\thread{E_2[\eapp\closek x]})}}{\thread{E_1[\eunit]}\PAR\thread{E_2[\eunit]}}}

    \infer{\procred{p_1}{p_2}}{\procred{p_1\PAR q}{p_2\PAR q}}

    \infer{\procred{p_1}{p_2}}{\procred{\NU xy{p_1}}{\NU xy{p_2}}}

    \infer{\isscong{p_1}{p_2} \\ \procred{p_2}{p_3} \\ \isscong{p_3}{p_4}}{\procred{p_1}{p_4}}
  \end{mathpar}
  \caption{Term and process reduction.}
  \label{fig:exp-reduction}
\end{figure}


Term and process reduction are in \cref{fig:exp-reduction}. Term reduction
comprises the standard axioms in System F with records, variants and recursion.
Evaluation contexts greatly simplify the structural reduction rules and pave the
way to process reduction. Following Milner~\cite{DBLP:journals/mscs/Milner92} we
factor out processes by means of a structural congruence relation that accounts
for the associative and commutative nature of parallel composition, scope
extrusion and exchanging the order of channel bindings. The rules closely follow
Almeida \etal and Gay and
Vasconcelos~\cite{DBLP:journals/corr/abs-2106-06658,DBLP:journals/jfp/GayV10}.
One finds axioms for forking new threads, creating new channels, for
communication ($\receivek$/$\sendk$, $\matchk$/$\selectk$ and
$\closek$/$\closek$), as well as structural rules to allow reduction underneath
parallel composition, channel creation and structural congruence.



We now address the metatheory of our language, starting with preservation for
both terms and processes.

\begin{theorem}[Preservation]\
  \begin{enumerate}
  \item If $\isterm \Delta \Gamma t T$ and $\expred t{t'}$, then
    $\isterm \Delta \Gamma {t'} T$.
  \item If $\isproc \Gamma p$ and $\isscong p{p'}$, then $\isproc \Gamma {p'}$.
  \item If $\isproc \Gamma p$ and $\procred p{p'}$, then $\isproc \Gamma {p'}$.
  \end{enumerate}
\end{theorem}


Progress for the term language is assured only when the typing context contains
channel endpoints only. When $\Delta$ is understood from the context we write
$\Gamma^{\skind}$ to mean that $\Gamma$ contains only types of kind $\skind$,
that is $\istype \Delta T \skind$ for all types $\typec T$ in $\Gamma$. Well
typed terms are values, or else they may reduce or are ready to reduce at the
process level. Reduction in the case of session operations---$\receivek$,
$\sendk$, $\matchk$, $\selectk$, $\closek$---is pending a matching counterpart.

\begin{theorem}[Progress for the term language]
  If $\isterm \Delta{ \Gamma^{\skind}} t T$, then $\termc{t}$ is a value,
  $\termc{t}$ reduces, or $\termc{t}$ is stuck in one of the following
  forms: $\termc{E[\eapp \forkk v]}$,
  $\termc{E[\etapp \newk T]}$,
  $\termc{E[\eapp {\eapp {\etapp {\etapp \receivek T} U} v}]}$,
  $\termc{E[\eapp {\etapp {\eapp {\etapp \sendk U} T} v} x]}$,
  $\termc{E[\ematch y{\erecord{l_i}{t_i}}]}$,
  $\termc{E[\eapp {(\selectc{l_j}{T})}x]}$, or
  $\termc{E[\eapp\closek x]}$.
\end{theorem}

In order to state our result on the absence of runtime errors we need a few
notions on the structure of terms and processes; here we follow
Almeida \etal~\cite{DBLP:journals/corr/abs-2106-06658}.
The \emph{subject} of an expression $\termc e$, denoted by $\subj e$, is $\termc
x$ in the following cases.
\begin{equation*}
  \eapp{\etapp{\etapp{\receivek}{T}}{U}}{x}
  \quad
  \eapp{\etapp{\eapp{\etapp{\sendk}{T}}{v}}{U}}{x}
  \quad
  \ematch{x}{t}
  \qquad
  \eapp{(\selectc{l_j}{T})}{x}
  \qquad
  \eapp \closek x
\end{equation*}

Two terms $\termc{e_1}$ and $\termc{e_2}$ \emph{agree} on channel $\termc{xy}$,
notation $\agree xy{e_1}{e_2}$, in the following cases (symmetric forms
omitted).
\begin{gather*}
  \agree xy
  {\eapp{\etapp{\etapp{\receivek}{T}}{U}}{x}}
  {\eapp{\etapp{\eapp{\etapp{\sendk}{V}}{v}}{W}}{y}}
  \qquad
  \agree xy {\eapp \closek x}{\eapp \closek y}
  \\
  \agree xy
  {\ematch{x}{\erecord{l_i}{t_i}_{i\in I}}}
  {\eapp{(\selectc{l_j}{T})}{y}}
  \quad j \in I
\end{gather*}

A closed process is a \emph{runtime error} if it is structural congruent to
some process that contains a subexpression or subprocess of one of the
following forms.
\begin{enumerate}
\item $\eapp vu$ where $\termc v$ is not a $\termc\lambda$ or a $\reck$,
  $\termc v \neq \etapp{\etapp{\receivek}{T}}{U}$,
  $\eapp{\etapp{\sendk}{T}}{u}$,
  $\etapp{\eapp{\etapp{\sendk}{T}}{w}}{U}$,
  $\selectc{l_j}{T}$, 
  $\closek$,
  $\forkk$;
\item $\etapp{v}{T}$ where $\termc v$ in not a $\termc\Lambda$,
  $\termc v \neq \receivek$,
  $\etapp \receivek U$,
  $\sendk$,
  $\etapp \sendk U$,
  $\newk$;
\item $\elet{l_i}{x_i}{v}{t}$ and $\termc v$ is not of the form
  $\erecord{l_i}{u_i}$;
\item $\ecase{v}{t}$ and $\termc v \neq \evariant{l_j}{u}{T}$ or
  $\termc t \neq \erecord{l_i}{u_i}_{\termc{i\in I}}$ with $j \notin I$;
\item\label{it:error-session} $\eapp{\etapp{\etapp{\receivek}{T}}{U}}{v}$ or
  $\eapp{\etapp{\eapp{\etapp{\sendk}{T}}{u}}U}{v}$ or
  $\ematch{v}{t}$ or
  $\eapp{(\selectc{l}{T})}{v}$ or
  $\eapp{\closek}{v}$ and
  $\termc v$ is not an endpoint $\termc x$;
\item\label{it:error-same-subj} 
  $\thread{E_1[e_1]} \PAR \thread{E_2[e_2]}$ and $\subj{e_1} =
  \subj{e_2}$;
\item\label{it:error-no-redex}
\sloppy{ $\NU{x}{y}{(\thread{E_1[e_1]} \PAR \thread{E_2[e_2]})}$ and
  $\subj{e_1} = \termc x$ and $\subj{e_2} = \termc y$ and
  $\neg\agree xy{e_1}{e_2}$.
}\end{enumerate}

The first five cases are standard to system $F$ with records and variants. The
support for session types and concurrency in the first two cases (term and type
application) are derived from the types of values for such operators
(\cref{fig:term-constants}). \Cref{it:error-session} addresses session operators
applied to non endpoints. \Cref{it:error-same-subj} is for two concurrent
session operators on the same channel end. Finally, \cref{it:error-no-redex} is
for mismatches on two session operations on two endpoints for the same channel.

\begin{theorem}[Safety]
  If $\isproc{\Gamma^{\skind}}{p}$, then $\termc{p}$ is not a runtime error.
\end{theorem}


An \emph{algorithmic typing system} can be easily extracted from the declarative
system for terms in \cref{fig:typing} via a bidirectional type system. The
system, formulated along the lines of Almeida \etal~\cite{DBLP:journals/corr/abs-2106-06658},
converts the non-syntax directed judgement $\isterm \Delta \Gamma t T$ into two
functional judgements: type synthesis
$\termsynth \Delta {\Gamma_1} t T {\Gamma_2}$ and type check
$\termagainst \Delta {\Gamma_1} t T {\Gamma_2}$. In the algorithmic type system,
context $\Gamma_2$ contains the unused part of context $\Gamma_1$
\cite{DBLP:journals/iandc/Vasconcelos12,walker:substructural-type-systems,DBLP:conf/forte/ZalakainD21}.



\section{Related Work}
\label{sec:related}

We briefly discuss work that is closer to ours.\smallskip

\noindent \textit{Equirecursion in system $F$. }
In first investigations on equirecursive types, the notion of type equivalence
is often formulated in a coinductive
fashion~\cite{DBLP:conf/popl/AmadioC91,DBLP:journals/fuin/BrandtH98,DBLP:conf/lics/ColazzoG99,DBLP:conf/icfp/GapeyevLP00,DBLP:conf/esop/Glew02}.
Two types are equivalent if they unroll into the same infinite tree. Whenever
this unrolling is the only type-level computation, such trees are regular,
enabling efficient decision procedures. Some authors have studied equirecursion
together with other notions of type-level computation. Solomon considers parameterized
type definitions, which correspond to higher-order
kinds~\cite{DBLP:conf/popl/Solomon78}. These implicitly correspond to $\lambda$-terms, 
since reduction occurs as types are allowed to call other types. Some
authors consider equirecursion in system $F_\omega$, with weaker or stronger
notions of
equality~\cite{DBLP:phd/de/Abel2007,DBLP:conf/tacs/BruceCP97,DBLP:conf/popl/CaiGO16,DBLP:journals/scp/Hinze02}.
Regarding equirecursion in system $F$, the model of
Cai \etal~\cite{DBLP:conf/popl/CaiGO16} is the closest to ours, and indeed our results up
to $\FMuAstDotOmega$ can be seen as a generalization of theirs.
However, Cai \etal depart from the usual setting by
allowing non-contractive types (which most authors forbid, including this work),
requiring a sort of infinitary lambda calculus. Moreover, this work further
extends additional equivalence properties by including session types with their
distinctive semantics, such as sequential composition and duality. 
\smallskip

\noindent \textit{Session type systems. }
Session types were introduced in the 90s by Honda \etal~\cite{DBLP:conf/concur/Honda93,DBLP:conf/esop/HondaVK98,DBLP:conf/parle/TakeuchiHK94}.
Equirecursion was the first approach used to construct infinite session types,
which often allows type equality to be interpreted according to a coinductive
notion of bisimulation~\cite{DBLP:journals/mscs/KozenS17}. In this vein, Keizer \etal~\cite{DBLP:conf/esop/KeizerB021} utilize coalgebras to represent session
types. Since the inception of session types,
there has been an interest in extending the theory to nonregular
protocols~\cite{DBLP:conf/europar/Puntigam99,DBLP:conf/europar/RavaraV97,DBLP:conf/soco/Sudholt05}.
Context-free session types emerged as a natural extension, as it still allowed
for practical type equality
algorithms~\cite{freest,DBLP:journals/corr/abs-1904-01284,DBLP:conf/tacas/AlmeidaMV20,DBLP:journals/corr/abs-2203-12877,DBLP:journals/toplas/Padovani19,DBLP:conf/icfp/ThiemannV16}.
Other approaches that go beyond regular session types include nested session
types~\cite{DBLP:conf/esop/DasDMP21} as well as 1-counter, pushdown and
2-counter session types~\cite{DBLP:conf/fossacs/GayPV22}. However, the
more expressive notions are not amenable to practical type equivalence
algorithms, just like the higher-order types present in our system
$\FMuOmega$. Polymorphism in session types has also been a topic of interest,
with or without
recursion~\cite{DBLP:conf/esop/CairesPPT13,DBLP:journals/corr/Dardha14,DBLP:journals/iandc/DardhaGS17,DBLP:journals/mscs/Gay08,griffith2016polarized}.
\smallskip

\noindent\textit{Dual type operator. }
This work is, to the best of our knowledge, the first that internalises duality as a type constructor. Other settings, such as the language Alms~\cite{tov2012practical}, consider duality for session types as a user-definable, not built in, type function. Our $\Dual$ is a type operator, not a type function. The difference is that a type function involves a type-level computation, which converges to a type written without dual. For example, in Alms we would have $\tkeyword{dual}(\OUT\Int\End) = \IN\Int\End$ (as a type-level computation), both sides being the \emph{same} type. In our setting, $\tdual{(\semit{\OUTn\Int}\End)}$ is a type on its own, which happens to be equivalent to $\semit{\INn\Int}\End$. At the same time, our setting allows for types such as $\tdual\alpha$, or $\semit{(\tdual\alpha)}{\semit{T_1}{T_2}}$, which do not reduce.
\smallskip

\noindent\textit{Type equivalence algorithms. }
Algorithms for deciding the equivalence of types must inherently be related to
the computational power of the corresponding type system. This has been used
implicitly or explicitly to obtain decidability results. As already explained,
if equirecursion is the only type-level computation, types can be represented as
finite-state automata (or equivalently, infinite regular trees). Although some
exponential time algorithms were first
proposed~\cite{DBLP:journals/acta/GayH05}, it has been established that the
problem can be solved in quadratic time~\cite{DBLP:conf/tacas/LangeY16}, which
is to be expected as it matches the corresponding problem of bisimulation of
finite-state automata~\cite{hopcroftkarp:1971}; see also Pierce~\cite{DBLP:books/daglib/0005958}.

The next `simplest' model of computation is that of simple grammars, which
intuitively correspond to deterministic pushdown automata with a single
state~\cite{DBLP:conf/fossacs/GayPV22}.
Almeida \etal~\cite{DBLP:conf/tacas/AlmeidaMV20} provided a practical algorithm for checking the bisimilarity of simple grammars. By dropping the determinism assumption, we arrive at Greibach normal form grammars, which are equivalent to basic process algebras~\cite{DBLP:conf/parle/BaetenBK87,DBLP:journals/jacm/BaetenBK93}. Bisimilarity algorithms have been studied extensively in this setting~\cite{DBLP:conf/mfcs/BurkartCS95,DBLP:journals/iandc/ChristensenHS95,DBLP:journals/corr/abs-1207-2479,DBLP:journals/ipl/Kiefer13}; presently it is known that the complexity of the problem lies between EXPTIME and 2-EXPTIME, which does not exclude the possibility of a polynomial time algorithm for the simpler model of simple grammars.

In this paper we present a reduction from first-order grammars to $\FMuOmega$-types, showing that the more expressive type systems ($\FMuOmega$, presented here and in Cai \etal~\cite{DBLP:conf/popl/CaiGO16}, as well as its extensions) are at least as powerful as deterministic pushdown automata. 
As far as we know, the closest result to ours is by Solomon~\cite{DBLP:conf/popl/Solomon78}, which shows conversions between a universe of ``context-free types'' and deterministic context-free languages. The universe of types studied by Solomon is different from $\FMuOmega$. With some work we could prove that Solomon's types can be embedded into $\FMuOmega$, which would entail our result as a corollary. However, it is easier and simpler to prove directly the reduction as we did.

The equivalence problem for deterministic pushdown automata was a notorious open problem for a long time, until Sénizergues showed it to be decidable~\cite{DBLP:conf/icalp/Senizergues97,DBLP:conf/icalp/Senizergues02}. Since his proof, many authors have tried to refine the result in an attempt to arrive at an implementable algorithm~\cite{DBLP:journals/corr/abs-1010-4760,DBLP:journals/tcs/Stirling01,DBLP:conf/icalp/Stirling02}.
\smallskip

\noindent\textit{Concurrent term languages. }
The usefulness of a type system is directly related to its capability to be used
in a programming language. Type systems such as the ones discussed in this work
lend themselves quite readily to functional term
languages~\cite{DBLP:conf/icalp/ImNP13}. For session types, existing term
languages are either inspired in the pi
calculus~\cite{DBLP:journals/lmcs/DasP22,DBLP:journals/iandc/Vasconcelos12,DBLP:conf/ppdp/ToninhoCP11}
or in the lambda
calculus~\cite{DBLP:journals/jfp/GayV10,DBLP:conf/icfp/LindleyM16,DBLP:conf/esop/ToninhoCP13},
or even the two~\cite{DBLP:journals/toplas/ToninhoY21}. The system presented in
this paper is linear, meaning that resources must be used exactly
once~\cite{DBLP:journals/toplas/KobayashiPT99,walker:substructural-type-systems}.
Some authors go beyond linearity by considering unrestricted type
qualifiers~\cite{DBLP:conf/esop/KeizerB021,DBLP:journals/iandc/Vasconcelos12} or
manifest sharing~\cite{DBLP:journals/pacmpl/BalzerP17}.



\section{Conclusion and future work}
\label{sec:conclusion}


This paper introduces an extension of system $F$ which includes equirecursion,
lambda abstractions, and context-free session types. We present type equivalence
algorithms, and a term language and its metatheory. Although we have defined a
rather general system, it turns out that for practical purposes one must
restrict recursion to $\tmu\kast$, that is, to type-level monomorphic recursion.
In any case, the main system $\FMuAstSemiOmega$ is a non-trivial extension of
(the contractive fragment of) $\FMuAstOmega$ (studied by Cai
\etal~\cite{DBLP:conf/popl/CaiGO16}) as well as $\FMuSemi$ (studied by Almeida
\etal~\cite{DBLP:journals/corr/abs-2203-12877}).

We have only considered polymorphic types of a functional nature: type
$\foralltinfix\alpha\kind T$ must always be of kind $\tkind$. It is worth
investigating polymorphism over session types, as it would allow further
additional behaviour. For example, we could be interested in streaming values of
heterogeneous nature, as in type
$\tmuinfix\alpha\skind{\extchoice\records{\donel\colon\Skip,\morel\colon\semit{\foralltinfix\beta\tkind{\INn\beta}}{\alpha}}}$.
It is however unclear whether this extension would still allow a translation into a simple grammar.

We proved that the type equivalence problem for systems $\FMuOmega$, $\FMuDotOmega$, $\FMuSemiOmega$ is at least as hard as a non-efficiently-decidable problem. We conjecture that these systems have the same power as deterministic pushdown automata (and hence, admit decidable type equivalence), but we do not have a construction to prove this result. In any case, our proof that the type equivalence problem is at least as hard as the bisimilarity of deterministic pushdown automata is enough to justify focus on the significant fragment with restricted recursion.

We study either full recursion (for theoretical results) or recursion limited to
kind $\kast$ (for algorithmic results). It would be interesting to study
in-between kinds of recursion; the next natural example is
$\tmu{\karrow\kast\kast}$. What model of computation would we arrive at if we
consider types written with this recursion operator? We conjecture that types
$\FMuOmega$ and $\FMuDotOmega$, when restricted to recursion of kind
$\karrow\kast\kast$, would still be expressible as simple grammars, whereas such
a restriction in the more powerful $\FMuSemiOmega$ would take us beyond this
model, but perhaps without reaching the expressivity of deterministic pushdown
automata.



\bibliographystyle{splncs04}
\bibliography{references}

\begin{thebibliography}{10}
\providecommand{\url}[1]{\texttt{#1}}
\providecommand{\urlprefix}{URL }
\providecommand{\doi}[1]{https://doi.org/#1}

\bibitem{DBLP:phd/de/Abel2007}
Abel, A.: Type-based termination: a polymorphic lambda-calculus with sized
  higher-order types. Ph.D. thesis, Ludwig Maximilians University Munich
  (2007), \url{https://d-nb.info/984765581}

\bibitem{DBLP:journals/corr/abs-2106-06658}
Almeida, B., Mordido, A., Thiemann, P., Vasconcelos, V.T.: Polymorphic
  context-free session types. CoRR  \textbf{abs/2106.06658} (2021).
  \doi{10.48550/arXiv.2106.06658}

\bibitem{Almeida22}
Almeida, B., Mordido, A., Thiemann, P., Vasconcelos, V.T.: Polymorphic lambda
  calculus with context-free session types. Information and Computation
  (2022). \doi{10.1016/j.ic.2022.104948}

\bibitem{freest}
Almeida, B., Mordido, A., Vasconcelos, V.T.: {FreeST}, a programming language
  with context-free session types. \url{http://rss.di.fc.ul.pt/tools/freest/}
  (2019)

\bibitem{DBLP:journals/corr/abs-1904-01284}
Almeida, B., Mordido, A., Vasconcelos, V.T.: Freest: Context-free session types
  in a functional language. In: PLACES. {EPTCS}, vol.~291, pp. 12--23 (2019).
  \doi{10.4204/EPTCS.291.2}

\bibitem{DBLP:conf/tacas/AlmeidaMV20}
Almeida, B., Mordido, A., Vasconcelos, V.T.: Deciding the bisimilarity of
  context-free session types. In: TACAS. LNCS, vol. 12079, pp. 39--56. Springer
  (2020). \doi{10.1007/978-3-030-45237-7\_3}

\bibitem{DBLP:conf/popl/AmadioC91}
Amadio, R.M., Cardelli, L.: Subtyping recursive types. In: POPL. pp. 104--118.
  {ACM} Press (1991). \doi{10.1145/99583.99600}

\bibitem{DBLP:conf/parle/BaetenBK87}
Baeten, J.C.M., Bergstra, J.A., Klop, J.W.: Decidability of bisimulation
  equivalence for processes generating context-free languages. In: PARLE. LNCS,
  vol.~259, pp. 94--111. Springer (1987). \doi{10.1007/3-540-17945-3\_5}

\bibitem{DBLP:journals/jacm/BaetenBK93}
Baeten, J.C.M., Bergstra, J.A., Klop, J.W.: Decidability of bisimulation
  equivalence for processes generating context-free languages. J. {ACM}
  \textbf{40}(3),  653--682 (1993). \doi{10.1145/174130.174141}

\bibitem{DBLP:journals/pacmpl/BalzerP17}
Balzer, S., Pfenning, F.: Manifest sharing with session types. Proc. {ACM}
  Program. Lang.  \textbf{1}({ICFP}),  37:1--37:29 (2017).
  \doi{10.1145/3110281}

\bibitem{DBLP:books/daglib/0067558}
Barendregt, H.P.: The lambda calculus - its syntax and semantics, Studies in
  logic and the foundations of mathematics, vol.~103. North-Holland (1985)

\bibitem{barendregt1977type}
Barendregt, H.P.: The type free lambda calculus. In: Studies in Logic and the
  Foundations of Mathematics, vol.~90, pp. 1091--1132. Elsevier (1977)

\bibitem{DBLP:journals/fuin/BrandtH98}
Brandt, M., Henglein, F.: Coinductive axiomatization of recursive type equality
  and subtyping. Fundam. Informaticae  \textbf{33}(4),  309--338 (1998).
  \doi{10.3233/FI-1998-33401}

\bibitem{DBLP:conf/tacs/BruceCP97}
Bruce, K.B., Cardelli, L., Pierce, B.C.: Comparing object encodings. In: TACS.
  LNCS, vol.~1281, pp. 415--438. Springer (1997). \doi{10.1007/BFb0014561}

\bibitem{DBLP:conf/mfcs/BurkartCS95}
Burkart, O., Caucal, D., Steffen, B.: An elementary bisimulation decision
  procedure for arbitrary context-free processes. In: MFCS. LNCS, vol.~969, pp.
  423--433. Springer (1995). \doi{10.1007/3-540-60246-1\_148}

\bibitem{DBLP:conf/popl/CaiGO16}
Cai, Y., Giarrusso, P.G., Ostermann, K.: System {F}-omega with equirecursive
  types for datatype-generic programming. In: POPL. pp. 30--43. {ACM} (2016).
  \doi{10.1145/2837614.2837660}

\bibitem{DBLP:conf/esop/CairesPPT13}
Caires, L., P{\'{e}}rez, J.A., Pfenning, F., Toninho, B.: Behavioral
  polymorphism and parametricity in session-based communication. In: ESOP.
  LNCS, vol.~7792, pp. 330--349. Springer (2013).
  \doi{10.1007/978-3-642-37036-6\_19}

\bibitem{DBLP:journals/csur/CardelliW85}
Cardelli, L., Wegner, P.: On understanding types, data abstraction, and
  polymorphism. {ACM} Comput. Surv.  \textbf{17}(4),  471--522 (1985).
  \doi{10.1145/6041.6042}

\bibitem{DBLP:journals/iandc/ChristensenHS95}
Christensen, S., H{\"{u}}ttel, H., Stirling, C.: Bisimulation equivalence is
  decidable for all context-free processes. Inf. Comput.  \textbf{121}(2),
  143--148 (1995). \doi{10.1006/inco.1995.1129}

\bibitem{DBLP:conf/lics/ColazzoG99}
Colazzo, D., Ghelli, G.: Subtyping recursive types in kernel {Fun}. In: LICS.
  pp. 137--146. {IEEE} Computer Society (1999). \doi{10.1109/LICS.1999.782605}

\bibitem{DBLP:journals/corr/abs-2203-12877}
Costa, D., Mordido, A., Po{\c{c}}as, D., Vasconcelos, V.T.: Higher-order
  context-free session types in system {F}. In: PLACES. {EPTCS}, vol.~356, pp.
  24--35 (2022). \doi{10.4204/EPTCS.356.3}

\bibitem{CurryFeys58}
Curry, H.H., Feys, R., Craig, W. (eds.): Combinatory Logic, Volume I.
  North-Holland (1958)

\bibitem{DBLP:journals/corr/Dardha14}
Dardha, O.: Recursive session types revisited. In: BEAT. {EPTCS}, vol.~162, pp.
  27--34 (2014). \doi{10.4204/EPTCS.162.4}

\bibitem{DBLP:journals/iandc/DardhaGS17}
Dardha, O., Giachino, E., Sangiorgi, D.: Session types revisited. Inf. Comput.
  \textbf{256},  253--286 (2017). \doi{10.1016/j.ic.2017.06.002}

\bibitem{DBLP:conf/esop/DasDMP21}
Das, A., DeYoung, H., Mordido, A., Pfenning, F.: Nested session types. In:
  {ESOP}. LNCS, vol. 12648, pp. 178--206. Springer (2021).
  \doi{10.1007/978-3-030-72019-3\_7}

\bibitem{DBLP:journals/toplas/DasDMP22}
Das, A., DeYoung, H., Mordido, A., Pfenning, F.: Nested session types. {ACM}
  Trans. Program. Lang. Syst.  \textbf{44}(3),  19:1--19:45 (2022).
  \doi{10.1145/3539656}

\bibitem{DBLP:journals/lmcs/DasP22}
Das, A., Pfenning, F.: Rast: {A} language for resource-aware session types.
  Log. Methods Comput. Sci.  \textbf{18}(1) (2022).
  \doi{10.46298/lmcs-18(1:9)2022}

\bibitem{debruijn:1972:lambda}
De~Bruijn, N.G.: Lambda calculus notation with nameless dummies, a tool for
  automatic formula manipulation, with application to the {Church-Rosser}
  theorem. In: Indagationes Mathematicae. vol.~75, pp. 381--392. Elsevier
  (1972). \doi{10.1016/1385-7258(72)90034-0}

\bibitem{DBLP:conf/icfp/GapeyevLP00}
Gapeyev, V., Levin, M.Y., Pierce, B.C.: Recursive subtyping revealed:
  functional pearl. In: ICFP. pp. 221--231. {ACM} (2000).
  \doi{10.1145/351240.351261}

\bibitem{DBLP:conf/icfp/GauthierP04}
Gauthier, N., Pottier, F.: Numbering matters: first-order canonical forms for
  second-order recursive types. In: ICFP. pp. 150--161. {ACM} (2004).
  \doi{10.1145/1016850.1016872}

\bibitem{DBLP:journals/mscs/Gay08}
Gay, S.J.: Bounded polymorphism in session types. MSCS  \textbf{18}(5),
  895--930 (2008). \doi{10.1017/S0960129508006944}

\bibitem{DBLP:journals/acta/GayH05}
Gay, S.J., Hole, M.: Subtyping for session types in the pi calculus. Acta
  Informatica  \textbf{42}(2-3),  191--225 (2005).
  \doi{10.1007/s00236-005-0177-z}

\bibitem{DBLP:conf/fossacs/GayPV22}
Gay, S.J., Po{\c{c}}as, D., Vasconcelos, V.T.: The different shades of infinite
  session types. In: FoSSaCS. LNCS, vol. 13242, pp. 347--367. Springer (2022).
  \doi{10.1007/978-3-030-99253-8\_18}

\bibitem{DBLP:journals/corr/abs-2004-01322}
Gay, S.J., Thiemann, P., Vasconcelos, V.T.: Duality of session types: The final
  cut. In: PLACES. {EPTCS}, vol.~314, pp. 23--33 (2020).
  \doi{10.4204/EPTCS.314.3}

\bibitem{DBLP:journals/jfp/GayV10}
Gay, S.J., Vasconcelos, V.T.: Linear type theory for asynchronous session
  types. J. Funct. Program.  \textbf{20}(1),  19--50 (2010).
  \doi{10.1017/S0956796809990268}

\bibitem{girard1972interpretation}
Girard, J.Y.: Interpr{\'e}tation fonctionnelle et {\'e}limination des coupures
  de l'arithm{\'e}tique d'ordre sup{\'e}rieur. Ph.D. thesis, {\'E}diteur
  inconnu (1972)

\bibitem{DBLP:journals/tcs/Girard87}
Girard, J.: Linear logic. Theor. Comput. Sci.  \textbf{50},  1--102 (1987).
  \doi{10.1016/0304-3975(87)90045-4}

\bibitem{DBLP:conf/esop/Glew02}
Glew, N.: A theory of second-order trees. In: ESOP. LNCS, vol.~2305, pp.
  147--161. Springer (2002). \doi{10.1007/3-540-45927-8\_11}

\bibitem{griffith2016polarized}
Griffith, D.E.: Polarized substructural session types. Ph.D. thesis, University
  of Illinois at Urbana-Champaign (2016). \doi{10.2172/1562827}

\bibitem{DBLP:books/cu/HindleyS86}
Hindley, J.R., Seldin, J.P.: Introduction to Combinators and Lambda-Calculus.
  Cambridge University Press (1986)

\bibitem{DBLP:journals/scp/Hinze02}
Hinze, R.: Polytypic values possess polykinded types. Sci. Comput. Program.
  \textbf{43}(2-3),  129--159 (2002). \doi{10.1016/S0167-6423(02)00025-4}

\bibitem{DBLP:conf/concur/Honda93}
Honda, K.: Types for dyadic interaction. In: {CONCUR}. LNCS, vol.~715, pp.
  509--523. Springer (1993). \doi{10.1007/3-540-57208-2\_35}

\bibitem{DBLP:conf/esop/HondaVK98}
Honda, K., Vasconcelos, V.T., Kubo, M.: Language primitives and type discipline
  for structured communication-based programming. In: {ESOP}. LNCS, vol.~1381,
  pp. 122--138. Springer (1998). \doi{10.1007/BFb0053567}

\bibitem{hopcroftkarp:1971}
Hopcroft, J.E., Karp, R.M.: A linear algorithm for testing equivalence of
  finite automata. Tech. rep., Cornell University (1971)

\bibitem{DBLP:conf/icalp/ImNP13}
Im, H., Nakata, K., Park, S.: Contractive signatures with recursive types, type
  parameters, and abstract types. In: ICALP. LNCS, vol.~7966, pp. 299--311.
  Springer (2013). \doi{10.1007/978-3-642-39212-2\_28}

\bibitem{DBLP:journals/corr/abs-1010-4760}
Jan{\v c}ar, P.: Short decidability proof for {DPDA} language equivalence via
  1st order grammar bisimilarity. CoRR  \textbf{abs/1010.4760} (2010),
  \url{http://arxiv.org/abs/1010.4760}

\bibitem{DBLP:journals/corr/abs-1207-2479}
Jan{\v c}ar, P.: Bisimilarity on basic process algebra is in {2-ExpTime} (an
  explicit proof). Log. Methods Comput. Sci.  \textbf{9}(1) (2012).
  \doi{10.2168/LMCS-9(1:10)2013}

\bibitem{DBLP:conf/esop/KeizerB021}
Keizer, A.C., Basold, H., P{\'{e}}rez, J.A.: Session coalgebras: {A}
  coalgebraic view on session types and communication protocols. In: ESOP.
  LNCS, vol. 12648, pp. 375--403. Springer (2021).
  \doi{10.1007/978-3-030-72019-3\_14}

\bibitem{DBLP:journals/ipl/Kiefer13}
Kiefer, S.: {BPA} bisimilarity is {EXPTIME}-hard. Inf. Process. Lett.
  \textbf{113}(4),  101--106 (2013). \doi{10.1016/j.ipl.2012.12.004}

\bibitem{DBLP:journals/toplas/KobayashiPT99}
Kobayashi, N., Pierce, B.C., Turner, D.N.: Linearity and the pi-calculus. {ACM}
  Trans. Program. Lang. Syst.  \textbf{21}(5),  914--947 (1999).
  \doi{10.1145/330249.330251}

\bibitem{DBLP:conf/focs/KorenjakH66}
Korenjak, A.J., Hopcroft, J.E.: Simple deterministic languages. In: SWAT. pp.
  36--46. {IEEE} Computer Society (1966). \doi{10.1109/SWAT.1966.22}

\bibitem{DBLP:journals/mscs/KozenS17}
Kozen, D., Silva, A.: Practical coinduction. Math. Struct. Comput. Sci.
  \textbf{27}(7),  1132--1152 (2017). \doi{10.1017/S0960129515000493}

\bibitem{DBLP:conf/tacas/LangeY16}
Lange, J., Yoshida, N.: Characteristic formulae for session types. In: TACAS.
  LNCS, vol.~9636, pp. 833--850. Springer (2016).
  \doi{10.1007/978-3-662-49674-9\_52}

\bibitem{DBLP:conf/icfp/LindleyM16}
Lindley, S., Morris, J.G.: Talking bananas: structural recursion for session
  types. In: ICFP. pp. 434--447. {ACM} (2016). \doi{10.1145/2951913.2951921}

\bibitem{DBLP:journals/mscs/Milner92}
Milner, R.: Functions as processes. Math. Struct. Comput. Sci.  \textbf{2}(2),
  119--141 (1992). \doi{10.1017/S0960129500001407}

\bibitem{DBLP:journals/toplas/Padovani19}
Padovani, L.: Context-free session type inference. {ACM} Trans. Program. Lang.
  Syst.  \textbf{41}(2),  9:1--9:37 (2019). \doi{10.1145/3229062}

\bibitem{DBLP:books/daglib/0005958}
Pierce, B.C.: Types and programming languages. {MIT} Press (2002)

\bibitem{DBLP:conf/europar/Puntigam99}
Puntigam, F.: Non-regular process types. In: Euro-Par. LNCS, vol.~1685, pp.
  1334--1343. Springer (1999). \doi{10.1007/3-540-48311-X\_189}

\bibitem{DBLP:conf/europar/RavaraV97}
Ravara, A., Vasconcelos, V.T.: Behavioural types for a calculus of concurrent
  objects. In: Euro-Par. LNCS, vol.~1300, pp. 554--561. Springer (1997).
  \doi{10.1007/BFb0002782}

\bibitem{DBLP:conf/programm/Reynolds74}
Reynolds, J.C.: Towards a theory of type structure. In: Programming Symposium.
  LNCS, vol.~19, pp. 408--423. Springer (1974).
  \doi{10.1007/3-540-06859-7\_148}

\bibitem{DBLP:conf/icalp/Senizergues97}
S{\'{e}}nizergues, G.: The equivalence problem for deterministic pushdown
  automata is decidable. In: ICALP. LNCS, vol.~1256, pp. 671--681. Springer
  (1997). \doi{10.1007/3-540-63165-8\_221}

\bibitem{DBLP:conf/icalp/Senizergues02}
S{\'{e}}nizergues, G.: {L(A)} = {L(B)}? decidability results from complete
  formal systems. In: ICALP. LNCS, vol.~2380, p.~37. Springer (2002).
  \doi{10.1007/3-540-45465-9\_4}

\bibitem{DBLP:conf/popl/Solomon78}
Solomon, M.H.: Type definitions with parameters. In: POPL. pp. 31--38. {ACM}
  Press (1978). \doi{10.1145/512760.512765}

\bibitem{DBLP:journals/tcs/Stirling01}
Stirling, C.: Decidability of {DPDA} equivalence. Theor. Comput. Sci.
  \textbf{255}(1-2),  1--31 (2001). \doi{10.1016/S0304-3975(00)00389-3}

\bibitem{DBLP:conf/icalp/Stirling02}
Stirling, C.: Deciding {DPDA} equivalence is primitive recursive. In: ICALP.
  Lecture Notes in Computer Science, vol.~2380, pp. 821--832. Springer (2002).
  \doi{10.1007/3-540-45465-9\_70}

\bibitem{DBLP:conf/soco/Sudholt05}
S{\"{u}}dholt, M.: A model of components with non-regular protocols. In: SC.
  LNCS, vol.~3628, pp. 99--113. Springer (2005). \doi{10.1007/11550679\_8}

\bibitem{DBLP:conf/parle/TakeuchiHK94}
Takeuchi, K., Honda, K., Kubo, M.: An interaction-based language and its typing
  system. In: {PARLE}. LNCS, vol.~817, pp. 398--413. Springer (1994).
  \doi{10.1007/3-540-58184-7\_118}

\bibitem{DBLP:conf/icfp/ThiemannV16}
Thiemann, P., Vasconcelos, V.T.: Context-free session types. In: ICFP. pp.
  462--475. {ACM} (2016). \doi{10.1145/2951913.2951926}

\bibitem{DBLP:conf/ppdp/ToninhoCP11}
Toninho, B., Caires, L., Pfenning, F.: Dependent session types via
  intuitionistic linear type theory. In: PPDP. pp. 161--172. {ACM} (2011).
  \doi{10.1145/2003476.2003499}

\bibitem{DBLP:conf/esop/ToninhoCP13}
Toninho, B., Caires, L., Pfenning, F.: Higher-order processes, functions, and
  sessions: {A} monadic integration. In: ESOP. LNCS, vol.~7792, pp. 350--369.
  Springer (2013). \doi{10.1007/978-3-642-37036-6\_20}

\bibitem{DBLP:journals/toplas/ToninhoY21}
Toninho, B., Yoshida, N.: On polymorphic sessions and functions: {A} tale of
  two (fully abstract) encodings. {ACM} Trans. Program. Lang. Syst.
  \textbf{43}(2),  7:1--7:55 (2021). \doi{10.1145/3457884}

\bibitem{tov2012practical}
Tov, J.A.: Practical programming with substructural types. Ph.D. thesis,
  Northeastern University (2012)

\bibitem{DBLP:journals/iandc/Vasconcelos12}
Vasconcelos, V.T.: Fundamentals of session types. Inf. Comput.  \textbf{217},
  52--70 (2012). \doi{10.1016/j.ic.2012.05.002}

\bibitem{walker:substructural-type-systems}
Walker, D.: Advanced Topics in Types and Programming Languages, chap.
  Substructural Type Systems, pp. 3--44. The MIT Press (2005)

\bibitem{DBLP:conf/forte/ZalakainD21}
Zalakain, U., Dardha, O.: {\(\pi\)} with leftovers: {A} mechanisation in
  {Agda}. In: FORTE. LNCS, vol. 12719, pp. 157--174. Springer (2021).
  \doi{10.1007/978-3-030-78089-0\_9}

\end{thebibliography}

\appendix

\section{Proof of \texorpdfstring{\cref{thm:kindingdecidable}}{decidability of kinding}}
\label{app:kinding}

This section is devoted to proving \cref{thm:kindingdecidable}, that kinding is decidable. Since kinding requires normalisation (due to rule \rulename{K-TApp}), we must first investigate that problem.

Even without considering sequential composition and recursion, it is well-known that normalisation may not terminate ($\typec{(\lambda x.x\ x)\ (\lambda x.x\ x)}$ is the typical example); in fact it is undecidable whether a type normalises. 
The standard approach is to consider kinded types, which are strongly normalising. 
However, in our model kinding itself requires normalisation, which leads to a ``chicken-and-egg'' situation. The solution is to consider a notion of \emph{pre-kinding}. We will write $\prekind \Delta T \kind$ to mean that $\TT$ is pre-kinded with kind $\kind$. The rules for pre-kinding are the same as for kinding, with the exception of rule \rulename{K-TApp} which loses the normalisation proviso, becoming rule \rulename{PK-TApp}:
\begin{equation*}
\infer[\ktapp]{
\istype \Delta {T} {\karrow \kind {\kind'}} 
\premspace
\istype \Delta U \kind 
\premspace
\isnormalised {\tapp TU}}
{\istype \Delta {\tapp T U} {\kind'}}
\Rightarrow
\infer[\rulename{PK-TApp}]{
\prekind \Delta {T} {\karrow \kind {\kind'}} 
\premspace
\prekind \Delta U \kind}
{\prekind \Delta {\tapp T U} {\kind'}}
\end{equation*}

\begin{lemma}
\label{lem:prekinddecidable}
$\prekind \Delta T \kind$ is decidable (in linear time).
\end{lemma}
\begin{proof}
Pre-kinding uses rule \rulename{PK-TApp} instead of \rulename{PK-TApp}, and therefore does not require normalisation. Therefore, given $\Delta$ and $\TT$ a pre-kind can be inferred by traversing the abstract syntax tree defining $\TT$, and using the context $\Delta$ to infer the kind of variables.
This processes requires a single pass and thus terminates in linear time. Therefore $\prekind \Delta T \kind$ is also decidable in linear time.
\end{proof}

We can now use pre-kinding to look at normalisation. 
Recall that, by \cref{fig:reduction}, there are essentially four ways to reduce a type $\TT$: the usual $\beta$-reduction for an application of a $\lambda$-term (\rulename{R-$\beta$}); the reduction for the recursion operator (\rulename{R-$\mu$}); the reductions for sequential composition (\rulename{R-Seq1}, \rulename{R-Seq2} and \rulename{R-Assoc}); and the reductions for duals. 
Let us separate reductions arising from \rulename{R-$\mu$} from the other rules, \ie let us use $\muarrow$ for reduction under recursion and $\lseqarrow$ for the usual $\beta$-reduction as well as reductions for sequential composition and duals. 
Let us also extend these reductions to functions on an application, \ie we have $\mured{T\ V}{U\ V}$, $\lseqred{T\ V}{U\ V}$ resp. whenever $\mured TU$, $\lseqred TU$ resp.
In this way, we have $\betaarrow\ =\ \lseqarrow\cup\muarrow$. We also extend the notion of weak head normal form and normalisation into the reductions $\lseqarrow$ and $\muarrow$, writing $\lseqnormalred TU$ and $\munormalred TU$.

\begin{lemma}[Normalisation]
\label{lem:prekindstrongnormal}
If $\prekind \Delta T \kind$ then there exists a unique $\UT$ \st $\lseqnormalred TU$.
\end{lemma}

\begin{proof}
Straightforward extension of the normalisation result for simply-typed lambda calculus~\cite[Chapter 12]{DBLP:books/daglib/0005958}. The fact that we include reduction under sequential composition and duality does not invalidate the standard proof, since these reductions simplify the type by attempting to bring a type constructor to the front; in particular, the depth of the abstract syntax tree defining the type does not increase along such reductions.
\end{proof}

Notice that we cannot extend the above lemma to full $\beta$-normalisation since reduction for the recursion operator may increase the size of the resulting type. The simplest example is $\tmu\skind\ (\tabs\alpha\skind\alpha)$, which is pre-kinded as $\prekind {}{\tmu\skind\ (\tabs\alpha\skind\alpha)}{\skind}$ but does not normalise:
$$\mured{\tmu\skind\ (\tabs\alpha\skind\alpha)}{(\tabs\alpha\skind\alpha)\ (\tmu\skind\ (\tabs\alpha\skind\alpha))}\lseqred{}{\tmu\skind\ (\tabs\alpha\skind\alpha)}\mured{}{\cdots}$$

\kindingdecidable*

\begin{proof}
Let $\Delta$, $\TT$, $\kind$ be given. We can first determine whether $\prekind \Delta T \kind$ (\cref{lem:prekinddecidable}). 
If $\TT$ is not pre-kinded, then it is also not kinded. 
Otherwise we need to determine whether $\TT$ normalises, or equivalently, whether there is an infinite sequence of reductions $\TT = \betared{\typec{T_0}}{\typec{T_1}}\betared{}{\typec{T_2}}\betared{}{\cdots}$
By \cref{lem:prekindstrongnormal}, such a sequence would have to contain an infinite amount of $\mu$-reductions, where between two $\mu$-reductions there must be a finite number of the other reductions. In other words, we can construct (any finite prefix of) the sequence 
$$\TT = \lseqnormalred{\typec{T_0}}{\typec{T'_0}}\mured{}{\typec{T_1}}\lseqnormalred{}{\typec{T'_1}}\mured{}{\typec{T_2}}\lseqnormalred{}{\typec{T'_2}}\mured{}{\cdots}$$
where for each $i$, $\typec{T'_i}$ is necessarily reached from $\typec{T_i}$ after finitely many steps.

In the above sequence, the only possible reductions that can be applied to $\typec{T'_i}$ are $\mu$-reductions, since $\typec{T'_i}$ is whnf with respect to the other reductions. If $\typec{T'_i}$ does not have any $\mu$-reduction, the sequence terminates and we can correctly determine that $\TT$ normalises. Otherwise if $\typec{T'_i}$ (only) has a $\mu$-reduction, then it must fit into one of the following cases (this is where we use the restriction to recursion of kind $\kast$):
\begin{equation*}\begin{aligned}
\typec{T'_i} &= \mured{\tapp{\tmu\kast}\UT}{\tapp\UT{(\tapp{\tmu\kast}\UT)}}
\\
\typec{T'_i} &= \mured{\semit{(\tapp{\tmu\kast}\UT)}{V}}{\semit{(\tapp U{(\tapp{\tmu\kast}\UT)})}{V}}
\\
\typec{T'_i} &= \mured{\tdual{(\tapp{\tmu\kast}\UT)}}{\tdual{(\tapp U {(\tapp{\tmu\kast}\UT)})}}
\\
\typec{T'_i} &= \mured{\semit{(\tdual{(\tapp{\tmu\kast}\UT)})}{V}}{\semit{(\tdual{(\tapp U {(\tapp{\tmu\kast}\UT)})})}{V}}
\end{aligned}\end{equation*}
Note that, in each of the four cases above, the expression $\tmu\kast\ \UT$ reappears after the $\mu$-reduction (without change). Therefore, the number of different subexpressions $\tmu\kast\ \UT$ that might appear in the sequence of reductions is finite, and we can detect an infinite sequence by `tagging' the expression $\tmu\kast\ \UT$ and stopping once $\tmu\kast\ \UT$ reappears. Therefore, we can devise an algorithm for deciding whether $\TT$ normalises: follow the sequence of reductions, terminating: as soon as no reductions are possible, or as soon as we revisit a previously `tagged' $\tmu\kast\ \UT$.
\end{proof}



\end{document}